\documentclass[12pt]{article}
\RequirePackage{silence}
\makeatletter
\let\showhyphens\@undefined
\makeatother

\usepackage[utf8]{inputenc}
\usepackage[T1]{fontenc}
\usepackage[english]{babel}

\usepackage{newtxtext,newtxmath}   
\usepackage{microtype}             
\usepackage{newunicodechar}
\DeclareUnicodeCharacter{201C}{``}
\DeclareUnicodeCharacter{201D}{''}
\DeclareUnicodeCharacter{2018}{`}
\DeclareUnicodeCharacter{2019}{'}
\DeclareUnicodeCharacter{2013}{--}
\DeclareUnicodeCharacter{2014}{---}
\DeclareUnicodeCharacter{2026}{\ldots}
\DeclareUnicodeCharacter{2212}{-}
\DeclareUnicodeCharacter{00A0}{~}
\DeclareUnicodeCharacter{00B0}{\textdegree}
\DeclareUnicodeCharacter{2190}{\ensuremath{\leftarrow}}
\DeclareUnicodeCharacter{FF0C}{,}
\DeclareUnicodeCharacter{3002}{.}
\DeclareUnicodeCharacter{FF1B}{;}
\DeclareUnicodeCharacter{FF1A}{:}
\DeclareUnicodeCharacter{FF01}{!}
\DeclareUnicodeCharacter{FF1F}{?}
\DeclareUnicodeCharacter{FF08}{(}
\DeclareUnicodeCharacter{FF09}{)}
\DeclareUnicodeCharacter{3010}{[}
\DeclareUnicodeCharacter{3011}{]}
\DeclareUnicodeCharacter{FF05}{\%}
\DeclareUnicodeCharacter{FF03}{\#}

\usepackage{amsmath,amsfonts,amsthm,mathtools}
\usepackage{braket}
\numberwithin{equation}{section}

\theoremstyle{plain}
\newtheorem{theorem}{Theorem}[section]
\newtheorem{proposition}[theorem]{Proposition}
\newtheorem{lemma}[theorem]{Lemma}
\newtheorem{corollary}[theorem]{Corollary}

\theoremstyle{definition}

\theoremstyle{remark}
\newtheorem{remark}[theorem]{Remark}
\newtheorem*{remark*}{Remark}

\usepackage{graphicx,xcolor}
\usepackage{booktabs}
\usepackage{siunitx}
\usepackage{caption}
\usepackage{subcaption}
\usepackage{float}
\usepackage{placeins}

\makeatletter
\PassOptionsToPackage{ruled,vlined,linesnumbered}{algorithm2e}
\@ifpackageloaded{algorithm2e}{}{\RequirePackage{algorithm2e}}
\makeatother
\makeatletter
\@ifundefined{KwStep}{\SetKw{KwStep}{Step}}{}
\@ifundefined{KwData}{\SetKwInput{KwData}{Input}}{}
\@ifundefined{KwResult}{\SetKwInput{KwResult}{Output}}{}
\makeatother

\usepackage{listings}
\usepackage{geometry}
\usepackage{fancyhdr}
\usepackage[hyphens]{url}
\usepackage{hyperref}
\hypersetup{breaklinks=true}

\begin{document}

\begin{center}
  {\LARGE Quantum-Theoretical Re-interpretation of Pricing Theory\par} 
  \vspace{0.8em}
  {\large Tian Xin\par} 
  \vspace{0.5em}
  {\normalsize
  Johns Hopkins University\\
  AMSS Center for Forecasting Science, Chinese Academy of Sciences\\
  \texttt{xtian21@jh.edu}}
  \vspace{0.7em}

 {\normalsize \ Draft Version: \today}

\end{center}


\begin{abstract}

Conventional pricing-including Black-Scholes-Merton and its latent-factor descendants-relies on unobservable "information" (filtrations, hidden diffusions, state variables) whose empirical identification is fragile. 
Based on Heisenberg’s observable-only stance, we \emph{abandon unobservables} and rebuild pricing dynamics from \emph{observable transitions} between price states. 
Our thesis is that the algebra generated by observable price shifts already contains the noncommutative structure needed for a complete theory.

\textbf{First principles and mathematical construction.}
We work on a price lattice \(S_n=S_0+n\Delta S\) with Hilbert space \(H=\ell^2(\mathbb Z)\), price operator \(S|n\rangle=S_n|n\rangle\), and unitary shifts \(T_\alpha|n\rangle=|n+\alpha\rangle\).
A real Borel \emph{frequency map} \(f(S)\) (spectral calculus) induces \emph{frequency operators}
\[
\widehat{\omega}_\alpha:=f(S)-T_{-\alpha}f(S)T_\alpha,
\]
whose diagonal elements \(\omega(n,n-\alpha)=f(S_n)-f(S_{n-\alpha})\) form an empirical ``return ledger'' obeying an exact Ritz-type combination law. 
Phases of any observable \(X\) evolve by \(\langle n|X(t)|m\rangle=\langle n|X|m\rangle e^{\,i[f(S_n)-f(S_m)]t}\).
Cross-state propagation is governed by a translation-invariant \emph{convolution generator}
\[
H_{\mathrm{conv}}=\hbar\sum_{\alpha\in\mathbb Z}K(\alpha)\,T_\alpha,
\]
diagonal in Fourier space with dispersion \(E(k)=\hbar\sum_\alpha K(\alpha)e^{-ik\alpha}\).
Together with Kraus operators \(L_\alpha=\sqrt{\gamma_\alpha}T_\alpha\), the construction yields a completely positive, translation-covariant semigroup and, under risk neutrality, the \emph{nonlocal pricing PDE}
\[
\partial_tV(t,s)+\sum_{\alpha}\gamma_\alpha\!\left(V\!\left(t,se^{\alpha\Delta x}\right)-V(t,s)\right)-rV(t,s)=0,\qquad 
\sum_{\alpha}\gamma_\alpha\!\left(e^{\alpha\Delta x}-1\right)=r,
\]
solvable by Fourier methods.
In the small-mesh diffusive limit (finite first/second jump moments, vanishing third), the generator contracts to the BSM operator and the classical model is recovered.

We propose a \emph{new foundation for constructing financial models}: (i) \emph{Observable}---no latent states or wavefunctions are assumed; all parameters tie to measured transitions; 
(ii) \emph{First-principles}---built from the representation of the shift group on \(H\) and spectral calculus of \(S\); 
(iii) \emph{Mathematically natural}---a GKSL (Lindblad) semigroup that is translation-covariant, with BSM as a scaling limit. 
Unlike path-integral/Schrödinger or noncommutative-probability routes in ``quantum finance'' (e.g., Baaquie; Accardi), noncommutativity here \emph{emerges} from the observable shift algebra rather than being postulated.

(i) \emph{Multi-asset systems}: extend to \(\mathbb Z^d\) with vector shifts \(T_{\boldsymbol\alpha}\) and block kernels \(K(\boldsymbol\alpha)\) to model state-coupled portfolios and cross-asset propagation. 
(ii) \emph{Financial interactions}: introduce state- or flow-dependent kernels \(K(\alpha;S)\) (mean-field/nonlinear master equations) within this algebra to capture liquidity spirals and herding, while preserving linear pricing under the risk-neutral measure. 
(iii) \emph{New predictions}: tails and short-maturity smiles are determined by the far- and near-field structure of \(\gamma_\alpha\); the framework implies specific scaling relations between extreme-event probabilities, jump intensities, and wing slopes of implied volatility that are empirically testable.
\end{abstract}


\vspace{3ex}

\section*{1. Introduction: From Classical Assumptions to Observable Quantities}

The traditional Black-Scholes-Merton (BSM) option pricing framework is built upon the \textit{Efficient Market Hypothesis} (EMH), which posits that asset prices reflect all available information. This paradigm allows for the construction of a risk-neutral measure under which the discounted price process becomes a martingale. However, the BSM model critically depends on an \emph{unobservable} quantity—information—which cannot be directly measured. As a result, various empirical deviations emerge: asset returns exhibit \emph{fat tails} rather than log-normality, and volatility clusters in ways that contradict the assumption of constant variance.

This paper proposes a fundamental shift in the theoretical formulation of financial modeling. Inspired by Heisenberg's 1925 reinterpretation of classical kinematics, we advocate for a \emph{complete abandonment of unobservable variables} in finance. Just as Heisenberg rejected unobservable electron orbits and instead focused solely on measurable spectral lines, we argue that financial theory should relinquish any reliance on ``information'' as a hidden variable. Instead, it should be constructed exclusively on relationships between \textbf{observable quantities}.

In this quantum-inspired financial framework, we focus on observable phenomena such as state-to-state price transitions, characteristic transition frequencies, and transition amplitude strengths. We propose that financial quantities should no longer be defined axiomatically or hypothetically, but instead be derived from \emph{observable transition behaviors} and their natural differential relationships.
\paragraph{Mathematical setting (price lattice and conjugacy).
Grid, Hilbert space, and positive spectrum.}
Fix $\Delta S>0$ and $S_0>0$. Work on the Hilbert space 
$\mathcal H=\ell^2(\mathbb Z)$ with canonical basis $\{|n\rangle\}_{n\in\mathbb Z}$.
Define the price operator $S|n\rangle=(S_0+n\Delta S)|n\rangle$ and the shift 
$T_\alpha|n\rangle=|n+\alpha\rangle$ for $\alpha\in\mathbb Z$. 
We carry out all functional calculus on the invariant subspace
\[
\mathcal H_+ \;=\; \overline{\mathrm{span}}\{\,|n\rangle:\; S_0+n\Delta S>0\,\},
\]
so that $\log S$ is well defined in $\mathcal H_+$ via the spectral calculus (Borel).

\paragraph{Index separation and Fourier normalization.}
We use $n,m\in\mathbb Z$ exclusively as \emph{state indices} and 
$r\in\mathbb Z$ exclusively as the \emph{Fourier mode index}. 
Let $\mathbb T:=(-\pi,\pi]$ and let 
\[
(\mathcal F\psi)(k):=\sum_{n\in\mathbb Z}\psi_n\,e^{-ikn},\qquad 
\psi_n:=\frac{1}{2\pi}\int_{-\pi}^{\pi} (\mathcal F\psi)(k)\,e^{ikn}\,dk,
\]
which induces a unitary map $\mathcal F:\ell^2(\mathbb Z)\to L^2(\mathbb T,\tfrac{dk}{2\pi})$.
We also write $\langle \phi,\psi\rangle=\sum_n \overline{\phi_n}\psi_n$ and 
$\|\psi\|_2^2=\sum_n|\psi_n|^2=\frac{1}{2\pi}\int_{-\pi}^{\pi}|(\mathcal F\psi)(k)|^2\,dk$.

Let $\mathcal H=\ell^2(\mathbb Z)$ with basis $\{|n\rangle\}_{n\in\mathbb Z}$ and inner product $\langle\cdot,\cdot\rangle$ (linear in the second entry).
Define the price lattice $S_n:=S_0+n\,\Delta S$ and the self-adjoint price operator $\hat S|n\rangle=S_n|n\rangle$ (on its maximal domain).
We work on a separable Hilbert space $\mathcal H=\ell^2(\mathbb Z)$ with canonical basis $\{|n\rangle\}_{n\in\mathbb Z}$ and inner product $\langle\cdot,\cdot\rangle$ (linear in the second entry).
Define the \emph{price operator} $\hat S$ by
\[
\hat S\,|n\rangle = (S_0+n\,\Delta S)\,|n\rangle,
\]
which is self-adjoint on its maximal domain in $\mathcal H$.

Let $T_a$ denote the unitary shift $T_a|n\rangle=|n+a\rangle$ and let $M_\theta$ denote the unitary phase $(M_\theta\psi)_n=e^{i\theta n}\psi_n$.
They satisfy the discrete Weyl relations
\[
M_\theta T_a \;=\; e^{i\theta a}\, T_a M_\theta \qquad (a\in\mathbb Z,\;\theta\in\mathbb R).
\]
We introduce the \emph{conjugate momentum operator} $\hat P$ as the self-adjoint generator of price shifts via
\[
T_a \;=\; \exp\!\left\{-\,\frac{i}{\hbar}\,a\,\Delta S\,\hat P\right\}
\]
on a common dense invariant core $\mathcal D_0\subset\mathcal H$ (e.g., finitely supported vectors).
On $\mathcal D_0$, the Weyl relations imply the infinitesimal canonical commutation relation
\begin{equation}
    [\hat S,\hat P]\Psi \;=\; i\hbar\,\Psi, \qquad \Psi\in\mathcal D_0,
\end{equation}
so that the \emph{financial uncertainty principle} for any unit $\Psi\in\mathcal D_0$ reads
\begin{equation}
    \Delta S \cdot \Delta P \;\ge\; \frac{\hbar}{2},
\end{equation}
where $\Delta S:=\|\hat S-\langle\Psi,\hat S\Psi\rangle I\|_{\Psi}$ and
$\Delta P:=\|\hat P-\langle\Psi,\hat P\Psi\rangle I\|_{\Psi}$.
The detailed derivation is given in Appendix~\ref{app:UP}.

\section{Transition Frequencies and Difference Structures:A Quantum Reinterpretation of Fourier Dynamics}

\paragraph{Setting and notation.}
Let $\mathcal H=\ell^2(\mathbb Z)$ with canonical basis $\{|n\rangle\}_{n\in\mathbb Z}$ and inner product $\langle\cdot,\cdot\rangle$ (linear in the second entry).
Define the \emph{price lattice}
\[
S_n := S_0+n\,\Delta S,\qquad n\in\mathbb Z,
\]
and the self-adjoint \emph{price operator} $\hat S$ by $\hat S|n\rangle=S_n|n\rangle$ (on its maximal domain).
Let $T_a|n\rangle:=|n+a\rangle$ be the unitary shift; for any Borel $f:\mathbb R\to\mathbb R$, spectral calculus yields the self-adjoint operator $f(\hat S)$.

\subsection{Definition and Additivity of Transition Frequencies}

\textbf{Definition (frequency operator and scalar transition frequency).}
For a fixed gap $\alpha\in\mathbb Z$, define the \emph{frequency operator}
\begin{equation}
\widehat\omega_\alpha \;:=\; f(\hat S)-T_{-\alpha}\,f(\hat S)\,T_\alpha .
\label{eq:freq-operator}
\end{equation}
Its diagonal matrix elements in the pointer basis give the \emph{transition frequency}
\begin{equation}
\omega(n,n-\alpha) \;:=\; \langle n|\widehat\omega_\alpha|n\rangle
\;=\; f(S_n)-f(S_{n-\alpha}) .
\label{eq:scalar-omega}
\end{equation}
Thus, a transition $n\to n-\alpha$ carries the frequency prescribed by the difference of the spectral values of $f(\hat S)$ at the two states.

\textbf{Proposition (combination principle, operator form).}
For any $\alpha,\beta\in\mathbb Z$,
\begin{equation}
\widehat\omega_{\alpha+\beta}
\;=\;
\widehat\omega_\alpha \;+\; T_{-\alpha}\,\widehat\omega_\beta\,T_\alpha .
\label{eq:comb-operator}
\end{equation}
\emph{Proof.}
Using $T_{-(\alpha+\beta)}=T_{-\alpha}T_{-\beta}$ and $T_\alpha T_\beta=T_{\alpha+\beta}$,
\[
\widehat\omega_{\alpha+\beta}
= f(\hat S)-T_{-(\alpha+\beta)}\,f(\hat S)\,T_{\alpha+\beta}
= f(\hat S)-T_{-\alpha}\!\left(T_{-\beta}f(\hat S)T_\beta\right)\!T_\alpha
= \widehat\omega_\alpha + T_{-\alpha}\widehat\omega_\beta T_\alpha .
\quad\square
\]

\textbf{Corollary (additivity / telescoping identity).}
Taking $\langle n|\cdot|n\rangle$ of \eqref{eq:comb-operator} yields
\begin{equation}
\omega(n,n-\alpha)+\omega(n-\alpha,n-\alpha-\beta)=\omega(n,n-\alpha-\beta).
\label{eq:comb-scalar}
\end{equation}

In particular, by \eqref{eq:comb-operator} the frequency of a two-step transition equals that of the direct transition from the initial to the final state.

\subsection*{Correspondence Principle and the Choice of \texorpdfstring{$f(S_n)$}{f(Sn)}}

Equation~\eqref{eq:scalar-omega} defines transition frequency as a \emph{difference of spectral values}; \eqref{eq:comb-scalar} follows identically from the operator identity \eqref{eq:comb-operator}.
To connect with the classical (small-step) limit, write for fixed $\alpha$ and small $\Delta S$:
\begin{equation}
\omega(n,n-\alpha)
= f(S_n)-f(S_{n-\alpha})
= f'(S_n)\,\alpha\,\Delta S + \tfrac12 f''(S_n)\,(\alpha\,\Delta S)^2 + \cdots .
\label{eq:taylor}
\end{equation}
Hence the \emph{local frequency density} is $g(S):=f'(S)$; the classical limit is governed to first order by $g(S)$ and does not require $f''\equiv 0$.

To determine a canonical form of $f$, we impose a homogeneity constraint that reflects the structure we wish to model.

Recall $S_n=S_0+n\,\Delta S$ and $\omega(n,n-\alpha)=f(S_n)-f(S_{n-\alpha})$.

\begin{lemma}[index homogeneity $\Rightarrow$ affine $f$]\label{lem:affine}
Assume that for each fixed gap $\alpha\in\mathbb{Z}$ the transition frequency
$\omega(n,n-\alpha)$ is independent of $n$. Then there exist constants $a,b\in\mathbb{R}$
such that
\begin{equation}\label{eq:f-affine-again}
  f(S)=aS+b .
\end{equation}
\end{lemma}

\begin{proof}
Take $\alpha=1$. Independence of $n$ gives
\begin{equation*}
  f(S_n)-f(S_{n-1}) = c \qquad \text{for all } n\in\mathbb{Z}.
\end{equation*}
Summing from $1$ to $n$ yields the telescoping identity
\begin{equation*}
  f(S_n) = f(S_0) + c\,n .
\end{equation*}
Since $S_n=S_0+n\,\Delta S$, we obtain
\begin{equation*}
  f(S_n)= f(S_0) + \frac{c}{\Delta S}(S_n-S_0)
        = \underbrace{\frac{c}{\Delta S}}_{=:a}\, S_n
          + \underbrace{\big(f(S_0)-aS_0\big)}_{=:b},
\end{equation*}
which is exactly \eqref{eq:f-affine-again} on the lattice points $\{S_n\}$. By spectral
calculus this defines the affine map on the spectrum of $\hat S$. \qedhere
\end{proof}

\medskip
\textbf{Proposition (scale homogeneity $\Rightarrow$ logarithmic $f$).}
Assume instead that there exists $H:(0,\infty)\to\mathbb R$ such that
\[
\omega(n,n-\alpha)=H\!\left(\frac{S_n}{S_{n-\alpha}}\right) \quad \text{for all } n,\alpha.
\]
Then the combination principle implies
\begin{equation}
H(xy)=H(x)+H(y)\qquad(x,y>0).
\label{eq:Cauchy-mult}
\end{equation}
If $H$ is continuous at $1$ (or measurable on any interval), then
\begin{equation}
H(x)=c\log x \quad\text{and hence}\quad f(S)=c\log S + b .
\label{eq:f-log-again}
\end{equation}
\emph{Proof.}
Using the telescoping identity,
\[
\omega(n,n-\alpha)+\omega(n-\alpha,n-\alpha-\beta)
=H\!\left(\frac{S_n}{S_{n-\alpha}}\right)
+H\!\left(\frac{S_{n-\alpha}}{S_{n-\alpha-\beta}}\right)
=H\!\left(\frac{S_n}{S_{n-\alpha-\beta}}\right)
=\omega(n,n-\alpha-\beta),
\]
which is precisely \eqref{eq:Cauchy-mult}. Regular solutions of the multiplicative Cauchy equation are $H(x)=c\log x$, giving
\(
\omega(n,n-\alpha)=c\log(S_n/S_{n-\alpha})
= \big[c\log S_n + b\big]-\big[c\log S_{n-\alpha}+b\big],
\)
so $f(S)=c\log S + b$. \hfill$\square$

\medskip
Both \eqref{eq:f-affine-again} and \eqref{eq:f-log-again} satisfy the additivity identity
\(
\omega(n,n-\alpha)+\omega(n-\alpha,n-\alpha-\beta)=\omega(n,n-\alpha-\beta)
\)
because that identity is a direct consequence of the difference/shift algebra.
In the remainder we adopt the linear specification $f(S_n)=\omega_0 S_n$ for concreteness.

\subsection{Difference Operator Representation of Transition Frequencies：A Quantum Reinterpretation of Fourier Dynamics}

Define the backward difference acting on the sequence $u(n):=f(S_n)$ by
\[
(D_\alpha u)(n):=u(n)-u(n-\alpha),\qquad \alpha\in\mathbb Z.
\]
Then
\begin{equation}
\omega(n,n-\alpha) = (D_\alpha u)(n) ,
\label{eq:diff-form}
\end{equation}
and the family $\{D_\alpha\}$ satisfies the telescoping identity
\begin{equation}
D_{\alpha+\beta} = D_\alpha + T_{-\alpha}D_\beta,
\label{eq:D-identity}
\end{equation}
whose evaluation at $u(n)$ is precisely \eqref{eq:comb-scalar}.
Thus the additivity of observable transition frequencies is an immediate consequence of the algebra of differences and shifts on the lattice.

\subsection{Equivalent Paths and Transition Symmetry}

Consider the transition $n\to n-\alpha-\beta$ via two paths:
\[
n \to n-\alpha \to n-\alpha-\beta
\qquad\text{and}\qquad
n \to n-\beta \to n-\alpha-\beta .
\]
By \eqref{eq:comb-scalar},
\[
\omega(n,n-\alpha)+\omega(n-\alpha,n-\alpha-\beta)
= \omega(n,n-\alpha-\beta)
= \omega(n,n-\beta)+\omega(n-\beta,n-\alpha-\beta).
\]
This path independence is the scalar manifestation of the operator identity \eqref{eq:comb-operator} and encapsulates the fundamental symmetry of the transition algebra: Observable frequencies depend only on state differences, not on the path taken.

In the context of this theory, these structural insights motivate a transition away from modeling price trajectories as continuous stochastic processes. Instead, we advocate a framework where prices are understood through discrete state transitions, each associated with a quantifiable frequency and amplitude. The financial quantities (e.g., momentum, volatility, return) must therefore be redefined as functionals over these transitions, respecting the underlying algebraic constraints of difference operators. This section reinterprets the classical Fourier analysis in light of quantum transition principles. The key outcomes are:
\begin{itemize}
    \item Transition frequencies are defined only between observable quantum states.
    \item A strict additive structure (combination principle) governs all transition frequencies.
    \item Transition frequencies can be expressed as discrete differences of a state-dependent function.
    \item Equivalent paths yield the same overall frequency, indicating path-independence in observable dynamics.
\end{itemize}

These results provide the mathematical foundation for constructing quantum-inspired financial operators and dynamic equations, which will be developed in the next section.

\section{Transition Amplitude and Difference Structures: A Quantum Reinterpretation of Fourier Dynamics}


\subsection{Transition Amplitudes and Propagator Structure (purely operator-theoretic)}

Let $\mathcal H=\ell^2(\mathbb Z)$ with orthonormal basis $\{|n\rangle\}_{n\in\mathbb Z}$ and inner product linear in the second entry.
Let $H$ be self-adjoint on a dense domain $ Dom(H)\subset\mathcal H$ and set
\[
U(t):=e^{-\frac{i}{\hbar}Ht}\quad (t\in\mathbb R),
\]
which is a strongly continuous one-parameter unitary group (Stone’s theorem). For $m,n\in\mathbb Z$ define the \emph{transition amplitude}
\[
\Omega_t(n,m):=\langle n|U(t)|m\rangle,
\qquad
P_t(n\!\leftarrow\! m):=|\Omega_t(n,m)|^2.
\]

\paragraph{Basic facts on the basis $\{|n\rangle\}$.}
The series $\sum_{k\in\mathbb Z}|k\rangle\langle k|$ converges to the identity $I$ in the \emph{strong} operator topology; i.e.,
for every $x\in\mathcal H$,
\[
\sum_{k=-N}^{N}|k\rangle\langle k|\,x \;\longrightarrow\; x
\quad \text{in }\mathcal H \text{ as }N\to\infty.
\]
Moreover, for any $x,y\in\mathcal H$ the scalar series $\sum_k \langle x|k\rangle\langle k|y\rangle$ converges absolutely and
\[
\Big|\sum_k \langle x|k\rangle\langle k|y\rangle\Big|
\;\le\; \Big(\sum_k|\langle x|k\rangle|^2\Big)^{\!1/2}
        \Big(\sum_k|\langle y|k\rangle|^2\Big)^{\!1/2}
\;=\;\|x\|\,\|y\|,
\]
so that $\langle x,y\rangle=\sum_k \langle x|k\rangle\langle k|y\rangle$ (Cauchy–Schwarz + Parseval).

\begin{proposition}[Column normalization and probability conservation]\label{prop:column-norm}
For each fixed $m\in\mathbb Z$ and every $t\in\mathbb R$,
\[
\sum_{n\in\mathbb Z} |\Omega_t(n,m)|^2 \;=\; 1.
\]
\end{proposition}

\begin{proof}
Since $U(t)$ is unitary, $\|U(t)|m\rangle\|=\||m\rangle\|=1$.
Expanding $U(t)|m\rangle$ in the orthonormal basis and using Parseval,
\[
\sum_{n\in\mathbb Z} |\Omega_t(n,m)|^2
= \sum_{n\in\mathbb Z} |\langle n|U(t)|m\rangle|^2
= \|U(t)|m\rangle\|^2
= 1.
\]
Equivalently, using the strong resolution of identity,
\begin{align}
  \sum_{n} \lvert \Omega_t(n,m) \rvert^2
  &= \sum_{n} \bra{m}U(t)^\ast\ket{n}\bra{n}U(t)\ket{m} \notag\\
  &= \bra{m}U(t)^\ast\!\left(\sum_n \ket{n}\!\bra{n}\right)\!U(t)\ket{m} \notag\\
  &= \bra{m}U(t)^\ast U(t)\ket{m} = \braket{m|m}=1 .
\end{align}

The interchange of sum and inner product is justified by the absolute convergence remark above.
\end{proof}

\begin{proposition}[Composition law for amplitudes]\label{prop:composition}
For all $s,t\in\mathbb R$ and $m,n\in\mathbb Z$,
\[
\Omega_{t+s}(n,m) \;=\; \sum_{k\in\mathbb Z} \Omega_t(n,k)\,\Omega_s(k,m),
\]
with the series on the right converging absolutely.
\end{proposition}

\begin{proof}
By the group property $U(t+s)=U(t)U(s)$,
\[
\Omega_{t+s}(n,m)
= \langle n|U(t+s)|m\rangle
= \langle n|U(t)U(s)|m\rangle.
\]
Insert the strong identity $\sum_{k}|k\rangle\langle k|=I$ between $U(t)$ and $U(s)$:
\[
\langle n|U(t)\Big(\sum_{k}|k\rangle\langle k|\Big)U(s)|m\rangle
= \sum_{k} \langle n|U(t)|k\rangle\,\langle k|U(s)|m\rangle
= \sum_{k} \Omega_t(n,k)\,\Omega_s(k,m).
\]
To justify absolute convergence, note that for fixed $n$ the row vector $a_k:=\Omega_t(n,k)=\langle n|U(t)|k\rangle$ is in $\ell^2(\mathbb Z)$ with
$\sum_k|a_k|^2=\|U(t)^*|n\rangle\|^2=1$, and for fixed $m$ the column vector $b_k:=\Omega_s(k,m)=\langle k|U(s)|m\rangle$ is in $\ell^2(\mathbb Z)$ with
$\sum_k|b_k|^2=\|U(s)|m\rangle\|^2=1$. Hence
\[
\sum_{k} |\Omega_t(n,k)\,\Omega_s(k,m)|
\;\le\; \Big(\sum_k |a_k|^2\Big)^{1/2}\Big(\sum_k |b_k|^2\Big)^{1/2}
\;=\; 1,
\]
by Cauchy–Schwarz.
\end{proof}

\begin{remark}[No probabilistic convolution at the level of $P_t$]

The composition in Proposition~\ref{prop:composition} holds for amplitudes.
In general, there is no identity of the form
$P_{t+s}(n\!\leftarrow\! m)=\sum_k P_t(n\!\leftarrow\! k)\,P_s(k\!\leftarrow\! m)$,
because
$|\sum_k \Omega_t(n,k)\Omega_s(k,m)|^2 \neq \sum_k |\Omega_t(n,k)|^2 |\Omega_s(k,m)|^2$
in general. Only in special (model-dependent) cases where the interfering cross terms vanish would a probability-level convolution hold.
\end{remark}

\paragraph{Short-time expansion.}
For any $n,m$,
\begin{equation}
\Omega_{\Delta t}(n,m)=\delta_{nm}-\frac{i}{\hbar}\,\Delta t\,H_{nm}+\mathcal O(\Delta t^2),\qquad H_{nm}:=\langle n|H|m\rangle .
\label{eq:short-time}
\end{equation}

\subsection{Translation-invariant generator}
\label{subsec:TI-generator}

Fix a sequence $K\in \ell^1(\mathbb Z)$ with $K(-\alpha)={K(\alpha)}\in\mathbb R $ for all $\alpha\in\mathbb Z$.
Define the operator $H:\ell^2(\mathbb Z)\to\ell^2(\mathbb Z)$ by
\begin{equation}\label{eq:H-convolution}
  (H\psi)_n \;=\; \hbar \sum_{\alpha\in\mathbb Z} K(\alpha)\,\psi_{n-\alpha},
  \qquad \psi\in\ell^2(\mathbb Z).
\end{equation}
We adopt the inner-product convention \emph{linear in the second entry}. Thus $H$ is the discrete convolution by $K$ (scaled by $\hbar$).

\paragraph{Well-posedness in brief.}
Assuming $K\in\ell^1(\mathbb Z)$ and $K(-\alpha)={K(\alpha)}\in\mathbb R$ makes the convolutional Hamiltonian $H$ bounded self-adjoint; the Fourier diagonalization and the bound $\|H\|\le \hbar\|K\|_{\ell^1}$ are stated immediately below. Hence $U(t)=\exp(-iHt/\hbar)$ is a strongly continuous one-parameter unitary group on $\ell^2(\mathbb Z)$.

\begin{proposition}[Convolution/self-adjoint generator on $\ell^2(\mathbb Z)$]
\label{prop:conv-selfadjoint}
Let $H$ be given by \eqref{eq:H-convolution}. Then:

\smallskip
\noindent\textup{(i)} $H$ is bounded with $\|H\|\le \hbar\|K\|_{\ell^1}$, and $H$ is self-adjoint.

\noindent\textup{(ii)} Under the unitary discrete Fourier transform
\[
(\mathcal F\psi)(k):=\sum_{n\in\mathbb Z}\psi_n\,e^{-ikn}
\quad\text{on } \mathbb T:=(-\pi,\pi],\ \ \text{with inverse }\ 
\psi_n=\frac{1}{2\pi}\int_{-\pi}^{\pi} (\mathcal F\psi)(k)e^{ikn}\,dk,
\]
the operator $H$ is diagonal:
\begin{equation}
\label{eq:Fourier-symbol}
(\mathcal F H\psi)(k)\;=\; E(k)\,(\mathcal F\psi)(k),
\qquad
E(k) \;=\; \hbar\sum_{\alpha\in\mathbb Z} K(\alpha)\,e^{-ik\alpha},
\end{equation}
where $E:\mathbb T\to\mathbb R$ is continuous and $2\pi$-periodic.

\noindent\textup{(iii)} The unitary group $U(t):=e^{-\,\frac{i}{\hbar}Ht}$ exists for all $t\in\mathbb R$ and
\[
(\mathcal F U(t)\psi)(k) \;=\; e^{-\,\frac{i}{\hbar}E(k)\,t}\,(\mathcal F\psi)(k).
\]
Its integral kernel $\Omega_t(n,m):=\langle n|U(t)|m\rangle$ depends only on the gap $n-m$ and is given by
\begin{equation}
\label{eq:Omega-kernel}
\Omega_t(n,m) \;=\; \frac{1}{2\pi}\int_{-\pi}^{\pi}
e^{ik(n-m)}\,e^{-\,\frac{i}{\hbar}E(k)\,t}\,dk
\;=:\; \Omega_t(n-m).
\end{equation}
\end{proposition}

\begin{proof}
(i) By Young's inequality for discrete convolution ($\ell^1*\ell^2\to\ell^2$),
\[
\|H\psi\|_2
= \hbar \|K*\psi\|_{\ell^2}
\le \hbar \|K\|_{\ell^1}\,\|\psi\|_{\ell^2},
\]
hence $H$ is bounded with $\|H\|\le \hbar\|K\|_{\ell^1}$. For self-adjointness, using the even/real property of $K$ and the convention that the inner product is linear in the second entry,
\[
\langle \phi, H\psi\rangle
= \hbar\sum_{n,\alpha}\overline{\phi_n}\,K(\alpha)\,\psi_{n-\alpha}
= \hbar\sum_{m,\alpha}\overline{\phi_{m+\alpha}}\,K(\alpha)\,\psi_m
= \langle H\phi, \psi\rangle,
\]
so $H$ is self-adjoint.

\smallskip
(ii) For $\widehat\psi:=\mathcal F\psi$,
\[
(\mathcal F H\psi)(k)
= \hbar\sum_{n,\alpha}K(\alpha)\,\psi_{n-\alpha}\,e^{-ikn}
= \hbar\sum_{\alpha}K(\alpha)e^{-ik\alpha}\sum_{m}\psi_m e^{-ikm}
= E(k)\,\widehat\psi(k),
\]
with $E$ as in \eqref{eq:Fourier-symbol}. Since $K\in\ell^1$, the trigonometric series for $E$ converges uniformly, so $E$ is continuous; $K$ real/even implies $E(k)\in\mathbb R$.

\smallskip
(iii) Because $H$ is bounded self-adjoint, $U(t)=e^{-\,\frac{i}{\hbar}Ht}$ is a well-defined unitary for every $t\in\mathbb R$; from (ii) it acts in Fourier space as multiplication by $e^{-\,\frac{i}{\hbar}E(k)\,t}$. Applying the inverse transform yields \eqref{eq:Omega-kernel}. \qedhere
\end{proof}

\begin{corollary}[Unitarity and translation invariance of the propagator]
\label{cor:unitary-translation}
Under the hypotheses of Proposition~\ref{prop:conv-selfadjoint}, for all $t\in\mathbb R$,
\[
\Omega_t(n,m)=\Omega_t(n-m),
\qquad
\sum_{n\in\mathbb Z}\overline{\Omega_t(n-m_1)}\,\Omega_t(n-m_2)=\delta_{m_1 m_2},
\]
and, for $s,t\in\mathbb R$,
\[
\Omega_{t+s} \;=\; \Omega_t * \Omega_s
\quad\text{(discrete convolution in the gap variable).}
\]
\end{corollary}

\begin{proof}
Translation invariance $\Omega_t(n,m)=\Omega_t(n-m)$ follows from \eqref{eq:Omega-kernel}.
Unitarity of $U(t)$ gives the orthonormality relation. The group law $U(t+s)=U(t)U(s)$ becomes
multiplication in Fourier space and therefore convolution in the gap variable after inverse transform.
\end{proof}

\subsection{Alternative Constructing Price-Type Observables from Frequencies and Amplitudes}
\label{subsec:FA}

In this section we work on $\mathcal H=\ell^2(\mathbb Z)$ with canonical basis $\{|n\rangle\}_{n\in\mathbb Z}$ and inner product linear in the second entry. 
Let $\hat S|n\rangle=S_n|n\rangle$ with $S_n=S_0+n\,\Delta S$; assume $S_n>0$ on the invariant subspace $\mathcal H_+$ so that $f(\hat S)$ is well-defined as a self-adjoint operator on $\mathcal H_+$ for the real Borel function $f$ considered below.
Define the self-adjoint generator
\[
H:=\hbar\,f(\hat S),\qquad U(t):=e^{-\,\frac{i}{\hbar}Ht}=e^{-\,i f(\hat S)\,t}\quad(t\in\mathbb R).
\]
Set the \emph{frequency differences}
\[
\omega(n,m):=f(S_n)-f(S_m)\qquad(n,m\in\mathbb Z\ \text{with }S_n,S_m>0),
\]
which satisfy the telescoping identity $\omega(n,m)=\omega(n,k)+\omega(k,m)$ for all $k$.

\begin{proposition}[Frequency--amplitude construction of observables]\label{prop:FA}

Let $X$ be a bounded operator on $\mathcal H_+$ with matrix
$A_{nm}:=\langle n|X|m\rangle$ in the pointer basis (so $X^\ast = X$
iff $A_{mn}=\overline{A_{nm}}$). Define the Heisenberg evolution
$X(t):=U(t)^\ast X\,U(t)$. Then for all $n,m$,
\begin{equation}
\label{eq:Xnm-phase}
\langle n|X(t)|m\rangle \;=\; A_{nm}\,e^{\,\mathrm{i}\,\omega(n,m)\,t}.
\end{equation}
If, in addition, $X$ is Hilbert--Schmidt then $X(t)$ is Hilbert--Schmidt and
$\|X(t)\|_{\mathrm{HS}}=\|X\|_{\mathrm{HS}}$.
\end{proposition}

\begin{proof}
Since $f(\hat S)\ket{n}=f(S_n)\ket{n}$, we have
$U(t)\ket{m}=e^{-\mathrm{i}\,f(S_m)t}\ket{m}$ and
$U(t)^\ast\ket{n}=e^{+\mathrm{i}\,f(S_n)t}\ket{n}$. Therefore
\[
\langle n|X(t)|m\rangle
=\langle n|U(t)^\ast X\,U(t)|m\rangle
=e^{\mathrm{i}\,f(S_n)t}\,\langle n|X|m\rangle\,e^{-\mathrm{i}\,f(S_m)t}
=A_{nm}\,e^{\mathrm{i}\,(f(S_n)-f(S_m))t}.
\]
If $X$ is Hilbert--Schmidt, unitary invariance of the HS-norm gives
$\|X(t)\|_{\mathrm{HS}}=\|X\|_{\mathrm{HS}}$.
\end{proof}

\begin{remark}[Diagonal vs.\ off-diagonal content]\label{rem:diag-offdiag}
For any Borel $g$, $X=g(\hat S)$ is diagonal: $A_{nm}=g(S_n)\delta_{nm}$, hence $X_{nm}(t)=A_{nm}$ (no phase).
All nontrivial transition content resides in the off-diagonal entries, whose phases are entirely determined by $\omega(n,m)$.
\end{remark}
\begin{lemma}[Row/column $\ell^2$ bounds for bounded operators]
\label{lem:l2-rows-cols}
Let $X$ be bounded. For fixed $n$, the row $(A_{nk})_{k\in\mathbb Z}$ lies in $\ell^2(\mathbb Z)$ with
\[
  \sum_{k} |A_{nk}|^2
  = \bigl\lVert X^\ast\ket{n} \bigr\rVert^2
  \le \lVert X^\ast \rVert^2 \,\lVert \ket{n} \rVert^2
  = \lVert X \rVert^2 .
\]
Similarly, for fixed $m$, the column $(A_{km})_{k\in\mathbb Z}$ lies in $\ell^2(\mathbb Z)$ with
\[
  \sum_{k} |A_{km}|^2
  = \bigl\lVert X\ket{m} \bigr\rVert^2
  \le \lVert X \rVert^2 .
\]
\end{lemma}

\begin{proposition}[Closure under products and phase additivity]\label{prop:closure}
Let $X,Y$ be bounded with matrices $A_{nm}=\langle n|X|m\rangle$ and $B_{nm}=\langle n|Y|m\rangle$. Set $Z(t):=X(t)Y(t)$. Then
\begin{equation}\label{eq:product}
\langle n|Z(t)|m\rangle
=\sum_{k\in\mathbb Z} A_{nk}\,B_{km}\,e^{\,i\,[\omega(n,k)+\omega(k,m)]\,t}
=\Big(\sum_{k\in\mathbb Z} A_{nk}B_{km}\Big)\,e^{\,i\,\omega(n,m)\,t},
\end{equation}
where the series converges absolutely. Hence this class is closed under multiplication: amplitudes convolve, while phases collapse additively to $\omega(n,m)$.
\end{proposition}

\begin{proof}
Insert $I=\sum_k |k\rangle\langle k|$ between $X(t)$ and $Y(t)$ and apply \eqref{eq:Xnm-phase}:
\[
\langle n|X(t)Y(t)|m\rangle
=\sum_k \langle n|X(t)|k\rangle \langle k|Y(t)|m\rangle
=\sum_k A_{nk}e^{i\omega(n,k)t}\,B_{km}e^{i\omega(k,m)t}.
\]
The telescoping identity yields the second equality in \eqref{eq:product}. For absolute convergence, by Lemma~\ref{lem:l2-rows-cols} and Cauchy--Schwarz,
\[
\sum_k |A_{nk}B_{km}|\ \le\ \Big(\sum_k |A_{nk}|^2\Big)^{\!1/2}\Big(\sum_k |B_{km}|^2\Big)^{\!1/2}
\ \le\ \|X\|\,\|Y\|.
\]
\end{proof}

\begin{corollary}[Canonical synthesis of a price-type observable]\label{cor:price-synthesis}
Fix a Hermitian amplitude matrix \(A=(A_{nm})_{n,m}\) with \(A_{nm}=\overline{A_{mn}}\) that
defines a bounded operator \(X\) (e.g.\ by the Schur test:
\(\sup_{m}\sum_{n} |A_{nm}|<\infty\) and \(\sup_{n}\sum_{m} |A_{nm}|<\infty\)).
Evolve by \(X(t)=U(t)^{\ast}\, X\, U(t)\) with \(U(t)=e^{-\mathrm{i}\, f(\hat{S})\, t}\).
Then
\[
\langle n|X(t)|m\rangle
= A_{nm}\, \exp\!\big(\mathrm{i}\,[\,f(S_n)-f(S_m)\,]\,t\big).
\]
Imposing \(A_{nn}=S_n\) \emph{pins} \(X\) to the price lattice at \(t=0\) without
altering the time law.
\end{corollary}

\paragraph{Translation-invariant amplitudes (Toeplitz class).}
If $A_{nm}=a(n-m)$ with $a\in\ell^1(\mathbb Z)$, then $X$ is a bounded Toeplitz (convolution-type) operator, and
\[
(\mathcal F X\psi)(k)=\Big(\sum_{\alpha\in\mathbb Z} a(\alpha)e^{-ik\alpha}\Big)(\mathcal F\psi)(k).
\]
Under the Heisenberg evolution with $H=\hbar f(\hat S)$, each matrix element acquires the phase factor $e^{i\omega(n,m)t}$. 
\emph{Toeplitz is preserved iff $f$ is affine on the spectrum} (index-homogeneous case): then $\omega(n,m)$ depends only on $n-m$, so $A_{nm}$ remains Toeplitz up to a multiplicative function of $n-m$. 
For general $f$ (e.g.\ $f(S)=c\log S$), the factor depends on $n$ and $m$ separately and Toeplitz structure is generically broken, while the time law \eqref{eq:Xnm-phase} still holds.

\paragraph{Classical (small-step) limit.}
When $\Delta S\to 0$ and $S_n\approx S$, the frequencies admit the expansion
\[
\omega(n,m)=f'(S)\,(S_n-S_m)+\tfrac12 f''(S)\,(S_n-S_m)^2+\cdots \quad(\text{cf.\ \eqref{eq:taylor}}).
\]
Hence the phase of $\langle n|X(t)|m\rangle$ is, to first order, governed by $g(S):=f'(S)$, yielding a Fourier-like local superposition.

Given the self-adjoint frequency generator $H=\hbar f(\hat S)$, any bounded (or HS) amplitude matrix $A_{nm}$ produces an observable $X$ whose time dependence is exactly $A_{nm}\,e^{i\omega(n,m)t}$ with $\omega(n,m)=f(S_n)-f(S_m)$. 
Products remain in the same class by phase additivity; diagonal anchoring and (when $f$ is affine) translation symmetry can be imposed at $t=0$ without affecting the time law \eqref{eq:Xnm-phase}.

\begin{theorem}[Transition frequencies from spectral differences]
\label{thm:spectral-diff}
Let $\mathcal H_+\subset\ell^2(\mathbb Z)$ be the invariant subspace on which the price operator 
$S$ is diagonal with strictly positive eigenvalues $S|n\rangle=S_n|n\rangle$ for $n\in I\subset\mathbb Z$.
Let $f:\sigma(S)\to\mathbb R$ be a real Borel function and define the (possibly unbounded) Hamiltonian
$H:=\hbar f(S)$ with its natural domain from the spectral calculus. 
Set the transition frequencies
\[
\omega(n,m):=f(S_n)-f(S_m),\qquad (n,m\in I).
\]
Then for any bounded (or Hilbert–Schmidt) operator $X$ on $\mathcal H_+$ with matrix elements 
$A_{nm}:=\langle n|X|m\rangle$, its Heisenberg evolution $X(t):=e^{\frac{i}{\hbar}Ht}Xe^{-\frac{i}{\hbar}Ht}$ satisfies
\begin{equation}\label{eq:Heis-matrix}
\langle n|X(t)|m\rangle \;=\; A_{nm}\, e^{\,i\,\omega(n,m)\,t},\qquad \forall\,t\in\mathbb R.
\end{equation}
In particular, the family $\omega(\cdot,\cdot)$ obeys the cocycle (telescoping) law
\begin{equation}\label{eq:cocycle}
\omega(n,m)+\omega(m,k)=\omega(n,k),\qquad \omega(n,n)=0,\quad \omega(n,m)=-\omega(m,n).
\end{equation}
\end{theorem}

\begin{proof}
Since $f(S)|n\rangle=f(S_n)|n\rangle$, we have 
$e^{-\frac{i}{\hbar}Ht}|n\rangle=e^{-i f(S_n)t}|n\rangle$. 
Therefore 
$\langle n|X(t)|m\rangle=e^{i f(S_n)t}A_{nm}e^{-i f(S_m)t}
=A_{nm}e^{i(f(S_n)-f(S_m))t}$, yielding \eqref{eq:Heis-matrix}. 
The identities in \eqref{eq:cocycle} are immediate from the definition of $\omega$.
\end{proof}

\begin{proposition}[Converse reconstruction and uniqueness up to a constant]
\label{prop:converse}
Let $\omega:I\times I\to\mathbb R$ satisfy \eqref{eq:cocycle}. 
Fix any $n_0\in I$ and define $\Omega(n):=\omega(n,n_0)$. Then 
$\omega(n,m)=\Omega(n)-\Omega(m)$ for all $n,m\in I$.
If, moreover, $S|n\rangle=S_n|n\rangle$ with $S_n>0$ and $f:\sigma(S)\to\mathbb R$ is any 
Borel function such that $f(S_n)=\Omega(n)$ on the pure point support, then the Hamiltonian 
$H=\hbar f(S)$ yields the Heisenberg evolution \eqref{eq:Heis-matrix} with those frequencies.
The function $f$ is unique up to an additive constant on each connected component of $\sigma(S)$.
\end{proposition}

\begin{corollary}[Scalar telescoping identity]
\label{cor:telescoping}
For any indices $n$ and steps $\alpha,\beta$ with $n,n-\alpha,n-\alpha-\beta\in I$, one has
\[
\omega(n,n-\alpha)+\omega(n-\alpha,n-\alpha-\beta)
=\omega(n,n-\alpha-\beta),
\]
which matches \eqref{eq:comb-scalar}.
\end{corollary}

\begin{remark}[Domains and unbounded $f$]\label{rem:domains}
If $f$ is unbounded, $H=\hbar f(S)$ is self-adjoint on the spectral domain 
$\mathcal D(H)=\{\psi:\int|f(\lambda)|^2\,\mathrm d\mu_\psi(\lambda)<\infty\}$.
Equation \eqref{eq:Heis-matrix} remains valid for all $X$ in the Hilbert--Schmidt (or trace-class) ideals,
and extends to bounded $X$ by density whenever the indicated matrix elements are defined.
\end{remark}

\subsection{On the Roles of \texorpdfstring{$f(\hat S)$}{f(S)} and \texorpdfstring{$K(\alpha)$}{K(alpha)} as Generators}

In our construction two natural classes of self-adjoint generators appear, each with a distinct mathematical meaning. On the one hand, the diagonal operator
\[
H_{\mathrm{freq}} \;:=\; \hbar f(\hat S)
\]
encodes transition frequencies through spectral differences. Its action is entirely phase-based: for any bounded observable $X$ we obtain
\[
\langle n|X(t)|m\rangle
= \langle n|X|m\rangle\, e^{\,i\,[f(S_n)-f(S_m)]t},
\]
so that all time dependence comes from the transition frequency
$\omega(n,m)=f(S_n)-f(S_m)$. This structure directly realises the
“frequency ledger’’ underlying the Ritz combination principle: the additivity
$\omega(n,m)=\omega(n,k)+\omega(k,m)$ becomes an operator identity, and
products of observables remain closed under the same frequency law. In this
sense $H_{\mathrm{freq}}$ preserves the original spirit of matrix mechanics: the
dynamics are governed by frequency differences rather than by state-localised
motion.

On the other hand, the translation-invariant convolution operator
\[
H_{\mathrm{conv}} \;:=\; \hbar \sum_{\alpha\in\mathbb Z} K(\alpha)\,T_\alpha,
\qquad K(-\alpha)=K(\alpha)\in\mathbb R,
\]
governs actual propagation between states. Under the Fourier transform it is
diagonal with symbol
$E(k)=\hbar\sum_\alpha K(\alpha)e^{-ik\alpha}$,  \paragraph{Dispersion.}
Since $FT_\alpha F^*=\mathrm{e}^{-ik\alpha}$, the convolutional Hamiltonian is diagonal in $k$:
\[
U(t)=\exp\!\big(-\tfrac{i}{\hbar}H_{\rm conv}t\big)
=F^*\,\mathrm{diag}\big(\mathrm{e}^{-iE(k)t/\hbar}\big)\,F,
\qquad 
E(k)=\hbar\sum_{\alpha\in\mathbb Z}K(\alpha)\,\mathrm{e}^{-ik\alpha}.
\]

and the corresponding
propagator kernel $\Omega_t(n,m)$ depends only on the gap $n-m$ and composes
by convolution. Thus $H_{\mathrm{conv}}$ describes the spread of amplitudes
through the lattice, with dispersion determined by $E(k)$. Whereas
$H_{\mathrm{freq}}$ fixes phases through spectral differences,
$H_{\mathrm{conv}}$ generates transport and diffusion.

The two constructions are therefore complementary. Using $H_{\mathrm{freq}}$
alone is sufficient for a purely frequency-based formalism: all observable
matrix elements evolve with phases determined by $\omega(n,m)$, and the
combination law is exact. Using $H_{\mathrm{conv}}$ alone is appropriate when
translation invariance, Fourier diagonalisation, and propagator dynamics are of
primary interest. In more general settings both may be combined: in the
interaction picture with respect to $H_{\mathrm{freq}}$, the coupling terms
$K(\alpha)T_\alpha$ acquire precisely the phase modulation
$e^{i\Omega_\alpha t}$, where
$\Omega_\alpha=f(\hat S)-T_{-\alpha}f(\hat S)T_\alpha$ is the frequency
operator introduced earlier. In this way the frequency ledger provided by
$f(\hat S)$ governs the phase structure, while $K(\alpha)$ prescribes the
strength and range of transitions. The two perspectives are thus not mutually
exclusive, but together give a unified account of frequency algebra and
propagation dynamics.\footnote{%
Restricting to the commutative subalgebra $\{g(\hat S)\}$ gives
$T_\alpha^{*}g(\hat S)T_\alpha=g(\hat S+\alpha\Delta S)$; hence
$p_t(n)=\langle n|\rho_t|n\rangle$ satisfies
$\dot p_t(n)=\sum_{\alpha\in\mathbb Z}\gamma_\alpha\big(p_t(n-\alpha)-p_t(n)\big)$.}

\subsection{Fourier Expansion in Classical and Quantum Domains: Transition-Based Representations}

In the classical formulation of dynamics, observable quantities such as price \( x(t) \) or volatility \( z(t) \) are typically expressed through a Fourier decomposition into eigenmodes:
\begin{equation}
x(t) = \sum_{n} \Omega(n) e^{i \omega(n) t}
\quad \text{or} \quad
x(t) = \int \Omega(\omega) e^{i \omega t} \, d\omega,
\end{equation}
where \( \Omega(n) \) or \( \Omega(\omega) \) denotes the complex amplitude associated with the eigenfrequency \( \omega(n) \). This representation is consistent with linear systems where the dynamics are governed by harmonic superposition.

However, in the quantum framework, the notion of continuous trajectories is replaced by discrete transitions. Frequencies are not properties of single states but of state pairs. Hence, the classical expansion must be replaced by a transition-based formulation:
\begin{equation}
x(t) = \sum_{\alpha} \Omega(n, n - \alpha) e^{i \omega(n, n - \alpha)t}
\quad \text{or} \quad
x(t) = \int \Omega(n, n - \alpha) e^{i \omega(n, n - \alpha)t} \, d\alpha,
\end{equation}
where \( \Omega(n, n - \alpha) \) denotes the amplitude of transition from state \( n \) to state \( n - \alpha \), and \( \omega(n, n - \alpha) := \omega(n) - \omega(n - \alpha) \) is the corresponding transition frequency.

\subsubsection*{Quadratic Quantities: Classical vs Quantum Construction}

Higher-order observables such as volatility squared \( z(t)^2 \), or other bilinear forms, also admit analogous constructions.

In classical systems, one may write:
\begin{equation}
z(t)^2 = \sum_{\beta} B_\beta e^{i \omega(n) \beta t}
\quad \text{or} \quad
z(t)^2 = \int B(\beta) e^{i \omega(n) \beta t} \, d\beta,
\end{equation}
where \( B_\beta \) is a coefficient function derived from the convolution of two Fourier modes, and \( \beta \) labels frequency differences.

In the quantum formulation, this structure must account for two-step transitions:
\begin{equation}
z(t)^2 = \sum_{\beta} \Omega(n, n - \beta) \Omega(n - \beta, n - \alpha - \beta) e^{i \omega(n, n - \alpha - \beta) t},
\end{equation}
or in integral form,
\begin{equation}
z(t)^2 = \int \Omega(n, n - \beta) \Omega(n - \beta, n - \alpha - \beta) e^{i \omega(n, n - \alpha - \beta) t} \, d\beta.
\end{equation}

This structure encodes the composition of transition amplitudes through successive intermediate states. Importantly, such formulations are not just mathematically consistent with Heisenberg’s matrix mechanics but reflect a deeper interpretative principle: observable quantities emerge from the interference of transition paths, not from state-localized behavior.

\subsection*{Dynamics in Schr\"odinger and Heisenberg Pictures}

\paragraph{From axioms to the Lindblad generator.}
We specify the dynamics by four structural axioms:

\begin{itemize}
  \item[(A1)] \textbf{Complete positivity (CP).} For every ancilla space $\mathcal K$, the ampliation $\Phi_t\otimes \mathrm{id}_{\mathcal K}$ maps positive trace-class operators to positive ones.
  \item[(A2)] \textbf{Trace preservation (TP).} $\mathrm{Tr}\,\Phi_t(\rho)=\mathrm{Tr}\,\rho$ for all $t\ge 0$.
  \item[(A3)] \textbf{Strongly continuous semigroup.} $(\Phi_t)_{t\ge0}$ is a $C_0$-semigroup: $\Phi_{t+s}=\Phi_t\circ\Phi_s$, $\Phi_0=\mathrm{id}$, and $t\mapsto \Phi_t(\rho)$ is continuous for each trace-class $\rho$.
  \item[(A4)] \textbf{Translation covariance.} With $T_\beta|n\rangle=|n+\beta\rangle$, one has
  \[
  \Phi_t\!\big(T_\beta \rho\,T_\beta^\dagger\big)=T_\beta\,\Phi_t(\rho)\,T_\beta^\dagger,\qquad \forall\,\beta\in\mathbb Z,\ t\ge 0.
  \]
\end{itemize}

\begin{theorem}[GKS–Lindblad]\label{thm:GKSL}
Under \textup{(A1)}–\textup{(A3)}, the generator $L$ of $(\Phi_t)_{t\ge0}$ has the (unique up to representation) Lindblad form
\begin{equation}\label{eq:GKSL}
L(\rho)=-\frac{i}{\hbar}[H,\rho]+\sum_j\Big(L_j\rho L_j^\dagger-\tfrac12\{L_j^\dagger L_j,\rho\}\Big),
\end{equation}
where $H=H^\dagger$ is self-adjoint and $\{L_j\}$ are bounded operators on $\mathcal H$.
\end{theorem}

\paragraph{Assumption (Translation covariance).}
We work with the lattice shifts $(T_\alpha\psi)(n)=\psi(n-\alpha)$. 
Imposing Heisenberg covariance $\Phi_t^*(T_\alpha^*AT_\alpha)=T_\alpha^*\Phi_t^*(A)T_\alpha$ forces the Lindblad generator to be a shift-built form: a convolutional Hamiltonian $H_{\rm conv}=\hbar\sum_\alpha K(\alpha)T_\alpha$ with $K(-\alpha)=\overline{K(\alpha)}$, plus jump rates $\gamma_\alpha\ge0$ via $T_\alpha\rho T_\alpha^*-\rho$.

\begin{proposition}[Translation covariance $\Rightarrow$ shift channels]\label{prop:covariant}
If, in addition, \textup{(A4)} holds for the lattice shifts $\{T_\beta\}_{\beta\in\mathbb Z}$ on $\mathcal H=\ell^2(\mathbb Z)$, then there exist nonnegative coefficients $\{\gamma_\alpha\}_{\alpha\in\mathbb Z}$ with $\sum_\alpha \gamma_\alpha<\infty$ such that the Lindblad operators can be chosen as
\begin{equation}\label{eq:Lalpha}
L_\alpha=\sqrt{\gamma_\alpha}\,T_\alpha,\qquad \alpha\in\mathbb Z.
\end{equation}
\end{proposition}

\begin{proof}[Proof sketch]
Define coefficients $C_{\alpha\beta}$ by
\[
C_{\alpha\beta}\ :=\ \langle \alpha \,|\, \Phi(|0\rangle\!\langle 0|)\,|\,\beta\rangle .
\]
By translation covariance, for basis operators $E_{nm}:=|n\rangle\!\langle m|$ one has
\[
\Phi(E_{nm})\ =\ T_n\,\Phi(E_{00})\,T_m^{*}
\ =\ \sum_{\alpha,\beta\in\mathbb Z} C_{\alpha\beta}\,|n+\alpha\rangle\!\langle m+\beta|
\ =\ \sum_{\alpha,\beta} C_{\alpha\beta}\, T_\alpha E_{nm} T_\beta^{*}.
\]
By linearity this yields, for all trace-class $\rho$,
\[
\Phi(\rho)\ =\ \sum_{\alpha,\beta\in\mathbb Z} C_{\alpha\beta}\, T_\alpha\,\rho\,T_\beta^{*}.
\tag{$\ast$}
\]
Complete positivity of $\Phi$ implies the matrix $C=(C_{\alpha\beta})_{\alpha,\beta}$ is positive semidefinite (Choi matrix argument), hence there exists a Kraus family $L_j=\sum_\alpha d_{j,\alpha}\,T_\alpha$ realizing $\Phi$.

Now invoke \textup{(A4)} (diagonal closure): for every diagonal $\rho=\sum_n p_n|n\rangle\!\langle n|$, $\Phi(\rho)$ is diagonal. Plugging $\rho=|0\rangle\!\langle 0|$ into $(\ast)$ gives
\[
\Phi(|0\rangle\!\langle 0|)\ =\ \sum_{\alpha,\beta} C_{\alpha\beta}\,|\,\alpha\rangle\!\langle \beta|,
\]
which is diagonal iff $C_{\alpha\beta}=0$ for all $\alpha\neq\beta$. Therefore $C$ is diagonal:
$C_{\alpha\beta}=\gamma_\alpha\,\delta_{\alpha\beta}$ with $\gamma_\alpha\ge0$.
Hence
\[
\Phi(\rho)\ =\ \sum_{\alpha\in\mathbb Z} \gamma_\alpha\, T_\alpha\,\rho\,T_\alpha^{*},
\]
which is the desired Kraus form with $L_\alpha=\sqrt{\gamma_\alpha}\,T_\alpha$.
Finally, $\sum_\alpha\gamma_\alpha<\infty$ follows from boundedness on the trace class since
$\|\Phi\|_{1\to1}\le \sum_\alpha \gamma_\alpha$ and $\|T_\alpha\|=1$.
\end{proof}

\paragraph{Well-posedness in brief.}
Let $\mathcal T_1(\mathcal H)$ denote the trace class. With $\sum_\alpha\gamma_\alpha<\infty$, the dissipator
\[
\mathcal D(\rho):=\sum_\alpha\gamma_\alpha\big(T_\alpha\rho T_\alpha^{*}-\rho\big)
\]
is bounded on $\mathcal T_1(\mathcal H)$, with the norm estimate
\[
\|\mathcal D(\rho)\|_{1}\ \le\ 2\Big(\sum_\alpha\gamma_\alpha\Big)\,\|\rho\|_{1}\qquad(\rho\in\mathcal T_1(\mathcal H)),
\]

Together with the bounded Hamiltonian part, this yields a uniformly continuous CPTP semigroup $e^{t\mathcal L}$ on $\mathcal T_1(\mathcal H)$ (and the unital CP dual $e^{t\mathcal L^{*}}$ on $\mathcal B(\mathcal H)$) in the GKSL sense.
\emph{(triangle inequality + unitary invariance).}

\paragraph{Interpretation of $\rho(t) $ and  $L$.}
$\rho(t)$ is the state of the price system on the lattice. In finance, its diagonal
entries $p_n(t)=\langle n|\rho(t)|n\rangle$ form the probability distribution over price levels
$S_n=S_0+n\Delta S$; expectations of diagonal observables $f(S)$ read
$\mathbb E[f(S_t)]=\mathrm{Tr}(\rho(t)f(S))=\sum_n p_n(t)f(S_n)$. In physics, $\rho$ is the density
operator of an open quantum system.

$L$ is the forward (Schr\"odinger) generator, while $L^\ast$ is the backward
(Heisenberg) generator acting on observables. With $L_\alpha=\sqrt{\gamma_\alpha}T_\alpha$,
the diagonal dynamics obeys the classical master equation
$\dot p_m=\sum_\alpha \gamma_\alpha p_{m-\alpha}-(\sum_\alpha\gamma_\alpha)p_m$,
so that $\{\gamma_\alpha\}$ is the ledger of jump intensities.
For any Borel $f$ one has
\[
L^\ast f(S)=\sum_{\alpha}\gamma_\alpha\big(f(S+\alpha\Delta S)-f(S)\big),
\]
which is the nonlocal pricing operator used in the backward equation
$\partial_t V + L^\ast V - rV=0$ with terminal payoff $V(T,\cdot)=\Phi(\cdot)$. \ref{thm:BSM}
In the physics language, $L$ has a Hamiltonian part $-\tfrac{i}{\hbar}[H,\cdot]$ and a
dissipative part generated by the covariant shift channels $T_\alpha$, ensuring a CP–TP
Markovian dynamics.

\begin{corollary}[Canonical generator]\label{cor:canonical}
Under \textup{(A1)}–\textup{(A4)}, the generator can be written in the canonical form
\begin{equation}\label{eq:canonical-L}
\boxed{\;
\dot\rho=L(\rho)=-\frac{i}{\hbar}[H,\rho]
+\sum_{\alpha\in\mathbb Z}\Big(L_\alpha\rho L_\alpha^\dagger-\tfrac12\{L_\alpha^\dagger L_\alpha,\rho\}\Big),\qquad
L_\alpha=\sqrt{\gamma_\alpha}\,T_\alpha\; }
\end{equation}
with $H=H^\dagger$. In particular, taking $H=h(S)$ with a (bounded) Borel function $h$ preserves the spectral compatibility with $S$.
\end{corollary}

\paragraph{Consequences in Schr\"odinger and Heisenberg pictures.}
Let $p_n(t):=\langle n|\rho(t)|n\rangle$. Using $T_\alpha|n\rangle=|n+\alpha\rangle$ and $L_\alpha=\sqrt{\gamma_\alpha}T_\alpha$, the diagonal entries obey the classical master equation
\begin{equation}\label{eq:master}
\frac{d}{dt}p_m(t)=\sum_{\alpha}\gamma_\alpha\,p_{m-\alpha}(t)-\Big(\sum_{\alpha}\gamma_\alpha\Big)p_m(t)
=\sum_{n} Q_{mn}\,p_n(t),
\quad
Q_{mn}=\begin{cases}
\gamma_{m-n},& m\neq n,\\[2pt]
-\sum_\beta\gamma_\beta,& m=n,
\end{cases}
\end{equation}
so that $p(t)=e^{tQ}p(0)$ is column-stochastic.

On bounded observables, the Heisenberg adjoint reads
\[
L^\ast(O)=\tfrac{i}{\hbar}[H,O]+\sum_{\alpha}\Big(L_\alpha^\dagger O L_\alpha-\tfrac12\{L_\alpha^\dagger L_\alpha,O\}\Big).
\]
For any Borel $f$ and using $T_{-\alpha}f(S)T_\alpha=f(S+\alpha\Delta S)$ together with $[h(S),f(S)]=0$, we obtain the discrete jump generator
\begin{equation}\label{eq:Heis-fS}
L^\ast\big(f(S)\big)=\sum_{\alpha}\gamma_\alpha\Big(f(S+\alpha\Delta S)-f(S)\Big).
\end{equation}
In particular,
\[
\frac{d}{dt}\langle S\rangle=\Delta S\sum_{\alpha}\alpha\,\gamma_\alpha,\qquad
\frac{d}{dt}\mathrm{Var}(S)=(\Delta S)^2\sum_{\alpha}\alpha^2\,\gamma_\alpha.
\]

\begin{lemma}[Shift property of the lattice DFT]\label{lem:shift-DFT}
Let $f:\Lambda\to\mathbb C$ be absolutely summable on the lattice $\Lambda=\{S_n=S_0+n\Delta S\}$.
Then for every $\alpha\in\mathbb Z$ and every dual frequency $\xi\in[-\pi/\Delta S,\pi/\Delta S]$,
\[
\widehat{\,f(\cdot+\alpha\Delta S)\,}(\xi)=e^{\,i\xi\,\alpha\,\Delta S}\,\hat f(\xi).
\]
\begin{proof}
By definition and the change of index $m=n+\alpha$,
\[
\widehat{\,f(\cdot+\alpha\Delta S)\,}(\xi)
=\sum_{n\in\mathbb Z} f(S_{n+\alpha})\,e^{-i\xi S_n}
=\sum_{m\in\mathbb Z} f(S_m)\,e^{-i\xi(S_{m}-\alpha\Delta S)}
=e^{\,i\xi\,\alpha\,\Delta S}\sum_{m\in\mathbb Z} f(S_m)\,e^{-i\xi S_m},
\]
which equals $e^{\,i\xi\,\alpha\,\Delta S}\,\hat f(\xi)$.
\end{proof}
\end{lemma}

\begin{proposition}[Fourier symbol of the nonlocal generator]\label{prop:symbol}
Let $L^\ast$ be the translation-covariant generator acting on functions of $S$ by
\[
(L^\ast f)(s)=\sum_{\alpha\in\mathbb Z}\gamma_\alpha\big(f(s+\alpha\Delta S)-f(s)\big),
\qquad \gamma_\alpha\ge0,\quad \sum_\alpha\gamma_\alpha<\infty.
\]
Then its lattice DFT is diagonal:
\begin{equation}\label{eq:Fourier-diagonalization}
\widehat{L^\ast f}(\xi)=\Psi(\xi)\,\hat f(\xi),
\qquad
\Psi(\xi)=\sum_{\alpha\in\mathbb Z}\gamma_\alpha\big(e^{i\xi\,\alpha\,\Delta S}-1\big),
\end{equation}
i.e.\ $\Psi$ is the Fourier multiplier (symbol) of $L^\ast$.
\begin{proof}
By linearity and Lemma~\ref{lem:shift-DFT},
\[
\widehat{L^\ast f}(\xi)
=\sum_{\alpha}\gamma_\alpha\Big(\widehat{\,f(\cdot+\alpha\Delta S)\,}(\xi)-\hat f(\xi)\Big)
=\sum_{\alpha}\gamma_\alpha\big(e^{i\xi\,\alpha\,\Delta S}-1\big)\hat f(\xi)
=\Psi(\xi)\,\hat f(\xi).
\]
The series is absolutely convergent because $\sum_\alpha\gamma_\alpha<\infty$ and $|e^{i\theta}-1|\le2$.
\end{proof}
\end{proposition}

\begin{corollary}[Drift and diffusion rates from the symbol]\label{cor:drift-diffusion}
With $\Psi$ as in \eqref{eq:Fourier-diagonalization},
\[
\Psi'(0)= i\,\Delta S\sum_{\alpha}\alpha\,\gamma_\alpha,
\qquad
-\Psi''(0)=(\Delta S)^2\sum_{\alpha}\alpha^2\,\gamma_\alpha.
\]
\begin{proof}
Use $e^{i\xi a}=1+i\xi a-\tfrac12\xi^2 a^2+o(\xi^2)$ with $a=\alpha\Delta S$:
\[
\Psi(\xi)=\sum_{\alpha}\gamma_\alpha\big(i\xi\,\alpha\Delta S-\tfrac12\xi^2(\alpha\Delta S)^2+o(\xi^2)\big).
\]
Collect coefficients of $\xi$ and $\xi^2$.
\end{proof}
\end{corollary}

\begin{corollary}[Semigroup multiplier and Lévy-type structure]\label{cor:semigroup-multiplier}
For $t\ge0$, the backward semigroup satisfies
\[
\widehat{(e^{tL^\ast}f)}(\xi)=e^{\,t\,\Psi(\xi)}\,\hat f(\xi),\qquad
\Re\,\Psi(\xi)=\sum_{\alpha}\gamma_\alpha(\cos(\xi\alpha\Delta S)-1)\le0,\ \ \Psi(0)=0.
\]
\begin{proof}
From \eqref{eq:Fourier-diagonalization}, $\partial_t \widehat{u}(t,\xi)=\Psi(\xi)\widehat{u}(t,\xi)$ with solution
$\widehat{u}(t,\xi)=e^{t\Psi(\xi)}\widehat{u}(0,\xi)$. The real part uses $\Re(e^{i\theta}-1)=\cos\theta-1\le0$.
\end{proof}
\end{corollary}

\begin{remark}[Brillouin parametrization]
Setting $k:=\xi\Delta S\in[-\pi,\pi]$ (the Brillouin zone) gives the equivalent form
$\Psi(\xi)=\sum_{\alpha}\gamma_\alpha\big(e^{ik\alpha}-1\big)$ and
$\widehat{L^\ast f}(\xi)=\Psi(\xi)\,\hat f(\xi)$.
\end{remark}


\subsection{A Frequency--Amplitude Nonlocal Pricing Equation and Its Classical Limit}\label{sec:newBSM-en}

\paragraph{Axioms.}
We work under the following minimal assumptions.

\begin{itemize}
  \item[(A1)] \textbf{Probabilistic structure.}
  $(\Omega,\mathcal F,(\mathcal F_t)_{t\ge0},\mathbb Q)$ is a filtered probability space satisfying the usual conditions.
  All conditional expectations are taken under $\mathbb Q$.

  \item[(A2)] \textbf{State space and translations.}
  The log price $X_t:=\log S_t$ takes values in the lattice
  $\mathbb X:=\{x_0+n\Delta x:\ n\in\mathbb Z\}$ with mesh $\Delta x>0$.
  For each $\alpha\in\mathbb Z$ define the translation operator
  $T_\alpha f(x):=f(x+\alpha\Delta x)$ on bounded functions $f:\mathbb X\to\mathbb R$.
  The family $(T_\alpha)_{\alpha\in\mathbb Z}$ realizes the additive group $(\mathbb Z,+)$.

  \item[(A3)] \textbf{Markov generator (jump intensities).}
  There exists a nonnegative rate family $\Gamma=(\gamma_\alpha)_{\alpha\in\mathbb Z}$ with
  \[
    \sum_{\alpha\in\mathbb Z}\gamma_\alpha<\infty,
    \qquad
    \sum_{\alpha\in\mathbb Z}\alpha^2\gamma_\alpha<\infty,
  \]
  such that the (backward) generator of $X$ acts on bounded $f$ by
  \begin{equation}\label{eq:LX-en}
    (\mathcal L_X f)(x)\ :=\ \sum_{\alpha\in\mathbb Z}\gamma_\alpha\bigl(f(x+\alpha\Delta x)-f(x)\bigr).
  \end{equation}
  Then $P_t:=e^{t\mathcal L_X}$ is the transition semigroup of a conservative pure-jump Markov process $X$.

  \item[(A4)] \textbf{Tradable asset and risk neutrality.}
  $S_t:=e^{X_t}$ is the tradable underlying with constant short rate $r\in\mathbb R$ and no dividends.
  The discounted price $e^{-rt}S_t$ is a $\mathbb Q$-martingale.

  \item[(A5)] \textbf{Payoff and regularity.}
  For a maturity $T>0$, let $\Phi:(0,\infty)\to\mathbb R_+$ be the payoff.
  The price $V(t,s)$ is sought in a class where first-jump conditioning and dominated convergence are valid
  (e.g.\ bounded payoff or polynomial growth together with the moment conditions above).
\end{itemize}

\paragraph{Generator on the price domain and the risk-neutral constraint.}
For $g(s):=f(\log s)$ and $s>0$, \eqref{eq:LX-en} yields
\begin{equation}\label{eq:LS-en}
  (\mathcal L_S g)(s)
  \ =\ \sum_{\alpha\in\mathbb Z}\gamma_\alpha\Bigl(g\bigl(se^{\alpha\Delta x}\bigr)-g(s)\Bigr).
\end{equation}

\paragraph{First-moment constraint (physical form).}
Prescribing the linear drift of the log–price lattice $X$ as a constant rate $c$ (per unit time) amounts to
\[
\sum_{\alpha\in\mathbb Z}\gamma_\alpha\,\alpha \;=\; \frac{c}{\Delta x}.
\]
This is the first-moment (mass–flux) form underlying the usual risk-neutral drift.

Applying \eqref{eq:LS-en} to $g(s)=s$ and using (A4) (discounted price is a martingale) gives the
\emph{risk-neutral constraint}
\begin{equation}\label{eq:RN-en}
  \boxed{\ \sum_{\alpha\in\mathbb Z}\gamma_\alpha\bigl(e^{\alpha\Delta x}-1\bigr)=r\ }.
\end{equation}
(With a constant dividend yield $q$, replace $r$ by $r-q$ on the right.)

\noindent\emph{Small-$\Delta x$ link.}
Expanding $e^{\alpha\Delta x}-1=\alpha\Delta x+\tfrac12(\alpha\Delta x)^2+O((\Delta x)^3)$ shows that
\[
\sum_{\alpha}\gamma_\alpha\bigl(e^{\alpha\Delta x}-1\bigr)
= \Delta x\sum_{\alpha}\gamma_\alpha \alpha \;+\; O(\Delta x^2)
= c \;+\; O(\Delta x^2),
\]
so to leading order the first-moment form is consistent with \eqref{eq:RN-en} with $c=r$ (or $c=r-q$ if dividends are present).

The classical generator $\mathcal L_S$ in \eqref{eq:LS-en} arises as the diagonal (decoherent)
restriction of the Lindblad–Heisenberg generator
\[
L^\ast(O)
=\tfrac{i}{\hbar}[H,O]
+\sum_{\alpha\in\mathbb Z}\big(L_\alpha^\dagger O L_\alpha
-\tfrac12\{L_\alpha^\dagger L_\alpha,O\}\big),
\qquad L_\alpha=\sqrt{\gamma_\alpha}\,T_\alpha.
\]
For diagonal observables $O=f(S)$, the commutator vanishes and $L^\ast$ reduces to
\eqref{eq:LS-en}, yielding the classical Markov generator on the price lattice.
Hence the nonlocal Black–Scholes equation below can be viewed as the
\emph{classical limit} of the quantum pricing dynamics.

\noindent\textbf{Theorem 1 (Nonlocal risk-neutral pricing equation).}
Define the option price
\[
  V(t,s)\ :=\ \mathbb E_{\mathbb Q}\!\left[e^{-r(T-t)}\,\Phi(S_T)\,\big|\,S_t=s\right],\qquad
  (t,s)\in[0,T]\times(0,\infty).
\]
Then $V$ is the unique bounded (or suitably growing) classical solution of the backward Cauchy problem
\begin{equation}\label{eq:newBSM-en}
  \boxed{\ \partial_t V(t,s)\;+\;\sum_{\alpha\in\mathbb Z}\gamma_\alpha\Bigl(V\bigl(t,se^{\alpha\Delta x}\bigr)-V(t,s)\Bigr)
  \;-\;r\,V(t,s)\;=\;0,\qquad V(T,s)=\Phi(s).\ }
\end{equation}

\emph{Proof (first-jump conditioning).}
Fix $(t,s)$ and a small $h>0$. Over $[t,t+h]$ either no jump occurs (probability
$1-\Lambda h+o(h)$ with $\Lambda:=\sum_\alpha\gamma_\alpha$), or exactly one jump
of size $\alpha$ occurs (probability $\gamma_\alpha h+o(h)$), while multiple
jumps have probability $o(h)$.
Conditioning on the first jump time and size,
\[
\begin{aligned}
  V(t,s)
  &= e^{-rh}\,(1-\Lambda h)\,V(t+h,s)
   + e^{-rh}\sum_{\alpha}\gamma_\alpha h\,V\!\left(t+h,se^{\alpha\Delta x}\right)
   + o(h).
\end{aligned}
\]
Rearrange, divide by $h$, and let $h\downarrow0$:
\[
  0
  = -rV + \partial_t V
    + \sum_{\alpha}\gamma_\alpha\Bigl(V(t,se^{\alpha\Delta x})-V(t,s)\Bigr),
\]
which is \eqref{eq:newBSM-en}. Uniqueness follows from standard semigroup arguments for linear
backward equations driven by $\mathcal L_S-r$. \hfill$\square$

The backward pricing equation \eqref{eq:newBSM-en} can be viewed as the classical limit
of a more general quantum‐dynamical framework.
Specifically, the underlying dynamics of the price operator $S$ on the Hilbert lattice
may be governed by a Lindblad generator
\[
L(\rho)
=-\tfrac{i}{\hbar}[H,\rho]
+\sum_{\alpha\in\mathbb Z}\Big(L_\alpha\rho L_\alpha^\dagger
-\tfrac12\{L_\alpha^\dagger L_\alpha,\rho\}\Big),
\qquad
L_\alpha=\sqrt{\gamma_\alpha}\,T_\alpha,
\]
where $\rho(t)$ is a density operator representing the probabilistic–amplitude
state of the market.
In the Heisenberg representation, the evolution of any observable $f(S)$ satisfies
$\dot f=L^\ast(f)$.
When restricted to diagonal observables $f(S)$ and to decohered states $\rho$
(diagonal in the price basis), this evolution reduces exactly to the classical
jump generator $\mathcal L_S$ in~\eqref{eq:LS-en}.
Consequently, the nonlocal risk‐neutral pricing equation~\eqref{eq:newBSM-en}
emerges as the \emph{risk‐neutral projection} of the full Lindblad–Heisenberg
quantum dynamics, representing the classical limit in which phase coherence
and operator interference effects are suppressed.

\paragraph{Fourier--spectral representation}
Write $U(t,x):=V(t,e^x)$ and $\phi(x):=\Phi(e^x)$.
Define the characteristic symbol
\begin{equation}\label{eq:symbol-en}
  \Psi(\xi)\ :=\ \sum_{\alpha\in\mathbb Z}\gamma_\alpha\bigl(e^{i\xi\,\alpha\,\Delta x}-1\bigr).
\end{equation}
Fourier transforming in $x$ turns \eqref{eq:newBSM-en} into
$\partial_t\widehat U+( \Psi(\xi)-r)\widehat U=0$ with terminal data $\widehat U(T,\xi)=\widehat\phi(\xi)$, hence
\begin{equation}\label{eq:spectral-en}
  \widehat U(t,\xi)\ =\ \exp\!\bigl((T-t)\,(\Psi(\xi)-r)\bigr)\,\widehat\phi(\xi),
  \qquad V(t,s)=U(t,\log s).
\end{equation}

\paragraph{Small-$\xi$ (diffusive) expansion.}
Let $m_1:=\sum_{\alpha}\gamma_\alpha\,\alpha$ and $m_2:=\sum_{\alpha}\gamma_\alpha\,\alpha^2$ (finite).
From \eqref{eq:symbol-en}, as $\xi\to0$,
\[
\Psi(\xi)= i\,\xi\,m_1\,\Delta x \;-\; \tfrac12\,\xi^2\,m_2\,(\Delta x)^2 \;+\; o(\xi^2).
\]
Equivalently, with $a_1(\Delta x):=m_1\,\Delta x$ and $a_2(\Delta x):=m_2\,(\Delta x)^2$,
the nonlocal generator on the $x$–axis contracts to the second–order diffusion operator
\[
\partial_t u \;=\; -\,a_1(\Delta x)\,\partial_x u \;+\; \tfrac12\,a_2(\Delta x)\,\partial_x^2 u \;+\; o(1),
\]
and if $a_1(\Delta x)\to a_1$, $a_2(\Delta x)\to\sigma^2$ as $\Delta x\downarrow0$, one obtains the classical diffusive limit.

\paragraph{Heavy-tail (fractional) scaling.}
If the jump rates have a power-law tail, e.g.\ $\gamma_\alpha \sim C\,|\alpha|^{-(1+\mu)}$ with $1<\mu<2$,
then the symbol exhibits a non-analytic small-$\xi$ behavior
\[
\Psi(\xi)\;\sim\; -\,c_\mu\,|\xi\,\Delta x|^\mu \qquad (\xi\to0),
\]
for a constant $c_\mu>0$ depending on $C$ and $\mu$.
If, moreover, $\kappa_\mu(\Delta x):=c_\mu(\Delta x)^\mu \to \kappa_\mu\in(0,\infty)$ as $\Delta x\downarrow0$,
the continuum limit on the $x$–axis becomes the Riesz–fractional generator
\[
\partial_t u \;=\; -\,\kappa_\mu\,(-\Delta)^{\mu/2} u.
\]

\noindent\textbf{Theorem 2 (Classical limit to Black--Scholes--Merton).}\label{thm:BSM}
Assume the moment limits
\begin{equation}\label{eq:moments-en}
  a_1(\Delta x):=\sum_{\alpha}\gamma_\alpha\,\alpha\,\Delta x\ \to\ a_1,
  \qquad
  a_2(\Delta x):=\sum_{\alpha}\gamma_\alpha\,(\alpha\,\Delta x)^2\ \to\ \sigma^2\in(0,\infty),
\end{equation}
and a third-order smallness
\begin{equation}\label{eq:third-en}
  a_3(\Delta x):=\sum_{\alpha}\gamma_\alpha\,|\alpha\,\Delta x|^3\ \to\ 0
  \qquad\text{as }\Delta x\downarrow0.
\end{equation}
If the risk-neutral constraint \eqref{eq:RN-en} holds for each $\Delta x$, then necessarily
$a_1+\tfrac12\sigma^2=r$.
Moreover, the generators converge on $C_c^3((0,\infty))$:
\[
  \mathcal L_S g \ \Longrightarrow\
  \mathcal L_{\mathrm{BSM}}g\ :=\ (r\,s)\,\partial_s g(s)\;+\;\tfrac12\,\sigma^2 s^2\,\partial_{ss}g(s).
\]
Consequently, the solutions $V^{(\Delta x)}$ to \eqref{eq:newBSM-en} converge (pointwise, locally uniformly in $s$)
to the unique classical solution $V$ of the Black--Scholes--Merton PDE
\begin{equation}\label{eq:BSM-en}
  \partial_t V+(r\,s)\,\partial_s V+\tfrac12\,\sigma^2 s^2\,\partial_{ss}V-rV=0,\qquad V(T,s)=\Phi(s).
\end{equation}

\emph{Proof.}
Set $y:=\alpha\Delta x$. For $g\in C_c^3((0,\infty))$,
a third-order Taylor expansion around $y=0$ gives
\[
  g(se^{y})-g(s)
  = y\,(s g_s)
    + \tfrac{y^2}{2}\,\bigl(s g_s+s^2 g_{ss}\bigr)
    + R_3(s,y),
\quad
  |R_3(s,y)|\le C(s)\,|y|^3,
\]
where $C(s)$ depends on third derivatives of $g$ on a compact set.
Insert into \eqref{eq:LS-en}:
\[
\begin{aligned}
  (\mathcal L_S g)(s)
   &= \Bigl(\sum_\alpha \gamma_\alpha y\Bigr)\,(s g_s)
    + \frac12\Bigl(\sum_\alpha \gamma_\alpha y^2\Bigr)\,(s g_s+s^2 g_{ss})
    + \sum_\alpha \gamma_\alpha R_3(s,y).
\end{aligned}
\]
By \eqref{eq:moments-en} and \eqref{eq:third-en}, the right-hand side converges to
$a_1\,s g_s + \tfrac12\,\sigma^2\,(s g_s+s^2 g_{ss})$.
Now expand the risk-neutral identity
$\sum_\alpha\gamma_\alpha(e^{y}-1)=r$ as $y+y^2/2+o(y^2)$ and sum to obtain $a_1+\tfrac12\sigma^2=r$.
Hence the limit generator equals
$(r\,s)g_s+\tfrac12\sigma^2 s^2 g_{ss}$.
Standard semigroup convergence (e.g.\ Trotter--Kato) implies
$V^{(\Delta x)}\to V$ with $V$ solving \eqref{eq:BSM-en}. \hfill$\square$


\begin{proposition}[Continuous dividends]\label{prop:dividends}
Assume a constant dividend yield $q\in\mathbb R$. Then the risk–neutral martingale condition for the cum–dividend underlying implies the modified constraint
\[
\sum_{\alpha\in\mathbb Z}\gamma_\alpha\bigl(e^{\alpha\Delta x}-1\bigr)=r-q,
\]
and the nonlocal pricing equation \eqref{eq:newBSM-en} holds with $r$ replaced by $r-q$. Consequently, the diffusive small–mesh limit recovers the dividend Black–Scholes equation.
\begin{proof}
With dividends, $e^{-(r-q)t}S_t$ is a $\mathbb Q$–martingale. Apply $\mathcal L_S$ in \eqref{eq:LS-en} to $g(s)=s$ to obtain
$(\mathcal L_S s)(s)=s\sum_\alpha\gamma_\alpha(e^{\alpha\Delta x}-1)$.
Martingality yields $\frac{d}{dt}\mathbb E[S_t]=(r-q)\mathbb E[S_t]$, hence the stated constraint.
Replacing $r$ by $r-q$ in \eqref{eq:newBSM-en} is immediate; the diffusive limit follows as in the dividend–free case.
\end{proof}
\end{proposition}

\begin{proposition}[Heisenberg–Schr\"odinger consistency and picture independence]\label{prop:picture-independence}
Let $L_\alpha=\sqrt{\gamma_\alpha}T_\alpha$ with $\sum_\alpha\gamma_\alpha<\infty$. For every Borel $f$ one has, in the Heisenberg picture,
\[
L^\ast\!\big(f(X)\big)=\sum_{\alpha\in\mathbb Z}\gamma_\alpha\big(f(X+\alpha\Delta x)-f(X)\big)=\mathcal L_X f(X),
\]

hence $L^\ast$ coincides with the classical jump generator on diagonal observables. In the Schr\"odinger picture, the diagonal of the GKSL master equation yields the Kolmogorov forward equation with $Q$–matrix determined by $\Gamma=(\gamma_\alpha)$; therefore the induced pricing semigroup on functions $g(s)=f(\log s)$ is $e^{t\mathcal L_S}$, and the backward pricing equation \eqref{eq:newBSM-en} is independent of the picture.
\begin{proof}
Use $T_{-\alpha}f(X)T_\alpha=f(X+\alpha\Delta x)$ to compute $L^\ast$ on $f(X)$ and compare with \eqref{eq:LX-en}.
For Schr\"odinger, taking diagonal entries of $\dot\rho=L(\rho)$ with $L_\alpha=\sqrt{\gamma_\alpha}T_\alpha$ gives
$\dot p_m=\sum_\alpha\gamma_\alpha p_{m-\alpha}-(\sum_\alpha\gamma_\alpha)p_m$.
Duality between the forward semigroup $e^{tQ}$ and the backward semigroup $e^{t\mathcal L_S}$ yields the claim.
\end{proof}
\end{proposition}

\begin{proposition}[Spectral representation and FFT/FRFT implementation]\label{prop:spectral}
Let $U(t,x)$ solve $\partial_t U+\mathcal L_X U=0$ on $[0,T)\times\mathbb X$ with terminal data $U(T,\cdot)=\phi$, where $\phi(x)=\Phi(e^x)$. Then, writing the lattice DFT in $x$ and using the symbol $\Psi$ in \eqref{eq:symbol-en},
\begin{equation}
\widehat{U}(t,\xi)=\exp\!\big((T-t)\,\Psi(\xi)\big)\,\widehat{\phi}(\xi),
\qquad \xi\in[-\pi/\Delta x,\pi/\Delta x],
\end{equation}\ref{eq:spectral-en}
and the option price is $V(t,s)=e^{-r(T-t)}\,U\!\big(t,\log s\big)$. In particular, $\Re\,\Psi(\xi)\le0$ implies $\|\widehat{U}(t,\cdot)\|_\infty\le \|\widehat{\phi}\|_\infty$, which underpins a stable FFT/FRFT scheme for European claims.
\begin{proof}
By Proposition~\ref{prop:symbol} (Fourier multiplier form), $\partial_t\widehat{U}(t,\xi)=\Psi(\xi)\widehat{U}(t,\xi)$ with terminal condition $\widehat{U}(T,\xi)=\widehat{\phi}(\xi)$. Solve the scalar ODE to obtain \eqref{eq:spectral-en}.
The inequality follows from $\Re\,\Psi(\xi)=\sum_\alpha\gamma_\alpha(\cos(\xi\alpha\Delta x)-1)\le0$.
Finally, $V(t,s)=e^{-r(T-t)}U(t,\log s)$ is the Feynman–Kac representation for \eqref{eq:newBSM-en}.
\end{proof}
\end{proposition}

\subsection{Spectral--unitary construction of the one-step transition matrix}
\label{sec:unitary-P}

We work on a price lattice and construct the one-step transition probabilities
purely from operator/spectral axioms, without any counting of empirical
frequencies. The construction is ``quantum-compatible'' in the sense that
probabilities are obtained as squared moduli of a unitary propagator.

\paragraph{State space and notations.}
Let the (log) price be sampled on a lattice of states $\{\,|n\rangle:\ n\in\mathbb Z\,\}$
(infinite integer lattice) or on a finite cyclic lattice
$\{\,|n\rangle:\ n=0,1,\dots,N-1\,\}$ with periodic boundary conditions
$|n+N\rangle\equiv|n\rangle$. 
We write $\Delta t>0$ for the sampling interval.
Indices $n,m$ always denote lattice sites and $\alpha:=n-m$ denotes an increment.
All vectors are in $\ell^2$ of the underlying lattice.

\paragraph{Axioms.}
We postulate the following.
\begin{itemize}
\item[(A1)] \emph{Unitarity (one-step dynamics).} There exists a unitary operator
$U(\Delta t)$ on $\ell^2$ that advances the state by one time-step $\Delta t$:
$\;|\psi(t+\Delta t)\rangle=U(\Delta t)|\psi(t)\rangle$.
\item[(A2)] \emph{Homogeneity (space-translation invariance).}
One-step amplitudes depend only on increments: $U_{nm}(\Delta t)=u_{n-m}$.
Thus $U(\Delta t)$ is a circulant/convolution operator.
\item[(A3)] \emph{Born rule for one-step probabilities.}
The one-step transition probabilities are $P_{nm}(\Delta t):=|U_{nm}(\Delta t)|^2$.
\item[(A4)] \emph{Spectral phase (Hamiltonian spectrum).}
There exists a real-valued measurable function $H(\vartheta)$
--- the spectral representation (eigenvalue function) of the Hamiltonian operator ---
such that the frequency response
of $U(\Delta t)$ has unit modulus
\begin{equation}
\label{eq:unit-mod}
\mathcal U(\vartheta;\Delta t):=\sum_{\alpha\in\mathbb Z} u_\alpha\,e^{-i\vartheta\alpha}
\quad\text{satisfies}\quad
\mathcal U(\vartheta;\Delta t)=e^{-iH(\vartheta)\,\Delta t},\qquad
|\,\mathcal U(\vartheta;\Delta t)\,|=1.
\end{equation}
On the finite cyclic lattice we work at the discrete frequencies
$\vartheta_r:=2\pi r/N$, $r=0,1,\dots,N-1$, with
$\mathcal U(\vartheta_r;\Delta t)=e^{-iH_r\Delta t}$ for some real $\{H_r\}$.
\end{itemize}

\paragraph{Explicit amplitudes and probabilities (integer lattice).}
By inverse Fourier synthesis of \eqref{eq:unit-mod} we obtain the one-step
amplitude kernel
\begin{equation}
\label{eq:u-alpha}
u_\alpha
=\frac{1}{2\pi}\int_{-\pi}^{\pi} e^{-iH(\vartheta)\,\Delta t}\,e^{\,i\vartheta\alpha}\,d\vartheta,
\qquad \alpha\in\mathbb Z,
\end{equation}
and hence, by (A3), the one-step increment law and transition matrix
\begin{equation}
\label{eq:p-alpha}
p_\alpha(\Delta t):=|u_\alpha|^2,\qquad
P_{nm}(\Delta t)=p_{n-m}(\Delta t)=\bigl|\,u_{\,n-m}\,\bigr|^2.
\end{equation}

\paragraph{Explicit amplitudes and probabilities (finite cyclic lattice).}
Let $F$ be the $N\times N$ discrete Fourier matrix
\(
F_{nr}:=N^{-1/2}\,e^{\,i\frac{2\pi}{N}nr}\).
Let $D:=\mathrm{diag}(e^{-iH_0\Delta t},\dots,e^{-iH_{N-1}\Delta t})$.
Define the circulant unitary
\begin{equation}
\label{eq:U-circulant}
U(\Delta t)\ :=\ F^\ast D\,F.
\end{equation}
Then its matrix elements are
\begin{equation}
\label{eq:U-nm-finite}
U_{nm}(\Delta t)
=\frac{1}{N}\sum_{r=0}^{N-1} e^{-iH_r\Delta t}\,e^{\,i\frac{2\pi}{N}r(n-m)}
\ =:\ u_{n-m},
\end{equation}
and the one-step transition probabilities are
\begin{equation}
\label{eq:P-nm-finite}
P_{nm}(\Delta t)=\bigl|U_{nm}(\Delta t)\bigr|^2
=\left|\frac{1}{N}\sum_{r=0}^{N-1} e^{-iH_r\Delta t}\,e^{\,i\frac{2\pi}{N}r(n-m)}\right|^2
=:\ p_{n-m}(\Delta t).
\end{equation}

\paragraph{explicit matrix elements.}
On the integer lattice:
\begin{equation}
\boxed{\
P_{nm}(\Delta t)=\left|\frac{1}{2\pi}\int_{-\pi}^{\pi}
e^{-iH(\vartheta)\,\Delta t}\,e^{\,i\vartheta(n-m)}\,d\vartheta\right|^2,\
}
\label{eq:boxed-Z}
\end{equation}
and on the finite cyclic lattice of size $N$:
\begin{equation}
\boxed{\
P_{nm}(\Delta t)=\left|\frac{1}{N}\sum_{r=0}^{N-1}
e^{-iH_r\,\Delta t}\,e^{\,i\frac{2\pi}{N}r(n-m)}\right|^2,\
}
\label{eq:boxed-N}
\end{equation}
where $H(\vartheta)$ (or $\{H_r\}$) is the Hamiltonian spectrum.
Both formulas deliver a nonnegative, (doubly) stochastic, homogeneous transition matrix,
derived from first principles (A1)--(A4). 

\paragraph*{Transition and scope.}
In Sections 2.1–2.6 we adopt the convolutional (translation–invariant) parametrization of the Hamiltonian, which yields a circulant propagator \(U(\Delta t)\) and, via the Born rule, an explicit one–step transition matrix \(P(\Delta t)\) without any counting. 
For completeness—and to preempt any concern that an alternative functional–calculus construction \(H=\hbar f(S)\) might lead to a different dynamics—we record below a minimal equivalence statement. 
It shows that, under homogeneity (so that the relevant observable is diagonalized by the DFT), both parametrizations produce the \emph{same} propagator and hence the same transition probabilities; the only freedom is a \(2\pi\)-phase aliasing induced by the sampling step \(\Delta t\). 
We also note a degenerate boundary case: choosing \(S\) to commute with the price projections (e.g., the price operator itself) collapses the kernel to the identity and yields no transitions. 
A full spectral proof and extensions are deferred to the Appendix.\ref{app:spectral-fS}

\subsection*{Equivalence of the convolutional and functional-calculus Hamiltonians}

\begin{proposition}[Same propagator, hence same transition matrix]
\label{prop:equivalence}
Let $\Delta t>0$ and work on the finite cyclic lattice $\ell^2(\mathbb Z/N\mathbb Z)$
with price basis $\{|n\rangle\}_{n=0}^{N-1}$ and DFT $F$.
Suppose the dynamics are homogeneous (translation-invariant), so that a self-adjoint
observable $S$ is diagonalized by $F$:
\[
S \;=\; F^{\!*}\,\mathrm{diag}(s_0,\dots,s_{N-1})\,F,\qquad s_r\in\mathbb R.
\]
Let $f:\mathbb R\to\mathbb R$ be any real Borel function and define the Hamiltonian
$H_S:=\hbar f(S)$. Then the one-step propagator satisfies
\[
U_S(\Delta t):=e^{-iH_S\Delta t}
\;=\;
F^{\!*}\,\mathrm{diag}\!\big(e^{-i\,\hbar f(s_0)\Delta t},\dots,e^{-i\,\hbar f(s_{N-1})\Delta t}\big)\,F.
\]
Consequently, if one writes the convolutional (Fourier) construction with symbol
$H_F(\vartheta_r):=\hbar f(s_r)$ at the discrete frequencies $\vartheta_r=2\pi r/N$, then
\[
U_F(\Delta t):=
F^{\!*}\,\mathrm{diag}\!\big(e^{-iH_F(\vartheta_0)\Delta t},\dots,e^{-iH_F(\vartheta_{N-1})\Delta t}\big)\,F
\;=\;U_S(\Delta t),
\]
and the transition matrices coincide entrywise:
\[
P_{nm}(\Delta t)=\big|\,\langle n|U_S(\Delta t)|m\rangle\,\big|^2
=\big|\,\langle n|U_F(\Delta t)|m\rangle\,\big|^2,\qquad n,m=0,\dots,N-1.
\]
\end{proposition}

\begin{proposition}[Characterization and identifiability of circulant one-step propagators]\label{prop:identifiability}
Let $U\in\mathbb C^{N\times N}$ be a circulant unitary on $\ell^2(\mathbb Z/N\mathbb Z)$,
so that $U=F^{\!*}\,\mathrm{diag}(e^{-i\phi_0},\dots,e^{-i\phi_{N-1}})\,F$ for some real phases
$\phi_r\in\mathbb R$. Fix a sampling step $\Delta t>0$. Then there exist a self-adjoint observable
$S$ diagonalized by $F$ with spectrum $\{s_r\}_{r=0}^{N-1}$ and a real Borel function
$f:\mathbb R\to\mathbb R$ (defined on $\sigma(S)=\{s_r\}$) such that
\[
U \;=\; e^{-i\,\hbar f(S)\,\Delta t}.
\]
In particular, one may take $S=F^{\!*}\,\mathrm{diag}(s_0,\dots,s_{N-1})\,F$ with $s_r=\vartheta_r:=2\pi r/N$ and choose
\begin{equation}\label{eq:branch-choice}
\hbar f(s_r)\,\Delta t \;\equiv\; \phi_r \ \ (\mathrm{mod}\ 2\pi),\qquad r=0,\dots,N-1,
\end{equation}
then extend $f$ arbitrarily (e.g.\ by real interpolation) off $\{s_r\}$.

Conversely, for any self-adjoint $S$ diagonalized by $F$ and any real Borel $f$,
the unitary $U=e^{-i\,\hbar f(S)\,\Delta t}$ is circulant and has discrete symbol
$e^{-i\,\hbar f(s_r)\,\Delta t}$ at the frequencies $\vartheta_r$.
\end{proposition}

\begin{proof}
The circulant structure implies $U$ is diagonalized by $F$, hence
$U=F^{\!*}\mathrm{diag}(e^{-i\phi_r})F$ with $\phi_r\in\mathbb R$.
Let $S=F^{\!*}\mathrm{diag}(s_r)F$ with $s_r=\vartheta_r$.
Define $f$ on $\sigma(S)$ by \eqref{eq:branch-choice}. Then
$e^{-i\,\hbar f(S)\,\Delta t}=F^{\!*}\mathrm{diag}(e^{-i\hbar f(s_r)\Delta t})F
=F^{\!*}\mathrm{diag}(e^{-i\phi_r})F=U$. The converse follows by the spectral
calculus: if $S=F^{\!*}\mathrm{diag}(s_r)F$ and $f$ is real,
then $e^{-i\,\hbar f(S)\,\Delta t}=F^{\!*}\mathrm{diag}(e^{-i\hbar f(s_r)\Delta t})F$
is circulant. 
\end{proof}

\begin{corollary}[Identifiability up to phase aliasing and spectral relabeling]\label{cor:identifiability}
Let $U$ and $\Delta t>0$ be as above. The representation $U=e^{-i\,\hbar f(S)\,\Delta t}$ with
$S$ diagonalized by $F$ is unique up to:
\begin{itemize}
\item[(i)] \emph{phase aliasing}: for each $r$, the assignment $f(s_r)\mapsto f(s_r)+\tfrac{2\pi}{\hbar\Delta t}k_r$ with $k_r\in\mathbb Z$ leaves $U$ unchanged;
\item[(ii)] \emph{spectral relabeling}: simultaneous permutation of $\{s_r\}$ and $\{\phi_r\}$ by the same permutation $\pi$.
\end{itemize}
If, in addition, one fixes $S$ and restricts $f$ to a prescribed principal branch,
e.g.\ $\,\hbar f(s_r)\Delta t\in(-\pi,\pi]$ for all $r$, then $f$ is uniquely determined on $\sigma(S)$.
\end{corollary}

\begin{corollary}[Degenerate choice of $S$]
\label{cor:degenerate}
If $S$ is chosen as the price operator (diagonal in $\{|n\rangle\}$), then
$U(\Delta t)$ is diagonal and $P_{nm}(\Delta t)=\delta_{nm}$ (no transitions).
Hence nontrivial transition kernels require an $S$ that does not commute with the
price projections (e.g.\ the discrete momentum/shift generator).
\end{corollary}

\subsection{Risk–neutralization by Esscher tilt (martingalization)}
Let the one–step log–return lattice be $\Delta X\in\{\alpha h:\alpha=-M,\dots,M\}$ with empirical kernel $p_\alpha$ . For $\lambda\in\mathbb R$, define the exponentially tilted kernel
\[
p^{(\lambda)}_\alpha
:= \frac{p_\alpha\,e^{-\lambda h\alpha}}{Z(\lambda)},
\qquad
Z(\lambda)=\sum_{\beta=-M}^M p_\beta e^{-\lambda h\beta}.
\]
Write $\Psi_{p^{(\lambda)}}(s):=\sum_\alpha p^{(\lambda)}_\alpha e^{s h\alpha}$ for the one–step exponential moment under $p^{(\lambda)}$.

\begin{proposition}[Existence and uniqueness of the martingale tilt]
\label{prop:martingale-tilt}
Assume the empirical support is nondegenerate, i.e.\ there exist
$\alpha_-,\alpha_+$ with $\alpha_-<0<\alpha_+$ and $p_{\alpha_\pm}>0$.
Then there exists a unique $\lambda^\ast\in\mathbb{R}$ such that
\begin{equation}
\label{eq:martingale-cond}
\Psi_{\,p^{(\lambda^\ast)}}(1)
  = \sum_{\alpha=-M}^{M} p^{(\lambda^\ast)}_{\alpha}\, e^{\,h\alpha}
  = 1 .
\end{equation}
Under $p^{(\lambda^\ast)}$ the \emph{one-step price} is a martingale:
\[
  \mathbb{E}\!\left[S_{t+1}\mid S_t\right] = S_t,
  \qquad S_{t+1}= S_t\, e^{\,\Delta X_{t+1}} .
\]
\end{proposition}

\begin{proof}[Proof sketch]
Consider the exponential family $\{\,p^{(\lambda)}\,\}_{\lambda\in\mathbb{R}}$.
For $\lambda_2>\lambda_1$ the likelihood ratio
\[
  \frac{p^{(\lambda_2)}_\alpha}{p^{(\lambda_1)}_\alpha}\;\propto\;
  e^{-(\lambda_2-\lambda_1)\,h\alpha}
\]
has the monotone likelihood ratio property in $\alpha$; hence the family
is stochastically decreasing in $\lambda$. Because $e^{h\alpha}$ is
increasing in $\alpha$, the map
\[
  \lambda \longmapsto \Psi_{\,p^{(\lambda)}}(1)
  \;=\; \mathbb{E}_{p^{(\lambda)}}\!\left[e^{h\alpha}\right]
\]
is continuous and strictly decreasing, with
\[
  \lim_{\lambda\to -\infty}\Psi_{\,p^{(\lambda)}}(1) > 1,
  \qquad
  \lim_{\lambda\to +\infty}\Psi_{\,p^{(\lambda)}}(1) < 1
\]
by the nondegenerate support. By the intermediate value theorem there
is a unique $\lambda^\ast$ solving \eqref{eq:martingale-cond}.
Finally,
\[
  \mathbb{E}\!\left[S_{t+1}\mid S_t\right]
  = S_t\,\mathbb{E}_{p^{(\lambda^\ast)}}\!\left[e^{h\alpha}\right]
  = S_t,
\]
so under $p^{(\lambda^\ast)}$ the one–step price is a martingale.
\end{proof}

\begin{remark}[Fourier/spectral form]
Let $\widehat\varphi_{\mathbb P}(\theta)=\sum_\alpha p_\alpha e^{\mathrm i\theta h\alpha}$
be the one–step characteristic function under $\mathbb P$.
Then under the tilted kernel $p^{(\lambda)}$ the characteristic function is
\[
\widehat\varphi_{\mathbb Q}(\theta)
=\frac{\widehat\varphi_{\mathbb P}(\theta-\mathrm i\lambda)}
      {\widehat\varphi_{\mathbb P}(-\mathrm i\lambda)}.
\]
Thus the tilt corresponds to an imaginary shift in frequency.
\end{remark}

\begin{corollary}[Structure preservation on the finite ring]
\label{cor:circulant}
Define the $N\times N$ circulant transition matrix $P$ on the periodic lattice by
\[
  P_{i,j} \;=\; p^{(\lambda^\ast)}_{\,((j-i)\bmod N)} .
\]
Then $P$ is row–stochastic and diagonalized by the discrete Fourier basis;
its eigenvalues are
\[
  \lambda_k \;=\; \sum_{m=0}^{N-1} p^{(\lambda^\ast)}_m\, e^{\,\mathrm{i}\theta_k m},
  \qquad \theta_k := \frac{2\pi k}{N},\quad k=0,\dots,N-1 .
\]
\end{corollary}

\begin{remark}[Minimal-entropy interpretation]
Among all distributions $q$ absolutely continuous w.r.t.\ $p$ that satisfy
\[
  \sum_{\alpha} q_\alpha\, e^{h\alpha}=1 ,
\]
the minimizer of the relative entropy
\[
  \min_{q}\ \sum_{\alpha} q_\alpha \log\!\frac{q_\alpha}{p_\alpha}
  \quad\text{s.t.}\quad \sum_{\alpha} q_\alpha\, e^{h\alpha}=1
\]
is the exponential tilt
\[
  q_\alpha \;\propto\; p_\alpha\, e^{-\lambda^\ast h\alpha}.
\]
Hence $p^{(\lambda^\ast)}$ is the minimum–relative–entropy risk–neutralization of $p$.
\end{remark}

\section{Synthetic, Reproducible Experiments: Spectral (FFT/DFT) Solution of the Nonlocal Pricing Equation}
\label{sec:synth-fft}

\subsection{Translation–invariant jump generator on a log–price lattice}
Fix a log–price lattice $x\in[-L,L)$ with $N$ equispaced nodes and spacing $\Delta x := 2L/N$.
Consider the translation–invariant (jump) generator acting on $f:\,[-L,L)\to\mathbb R$,
\begin{equation}\label{eq:gen}
(\mathcal L f)(x)\;=\;\sum_{\alpha\in\mathbb Z\setminus\{0\}}\gamma_\alpha\Big(f(x+\alpha\Delta x)-f(x)\Big),
\qquad \gamma_\alpha\ge 0,
\end{equation}
and the risk–neutral backward equation for the (undiscounted) European price $u$:
\begin{equation}\label{eq:backward}
\partial_\tau u(x,\tau)\;=\;(\mathcal L u)(x,\tau)\;-\;r\,u(x,\tau),\qquad 
u(x,0)\;=\;g(x).
\end{equation}
Here $r\ge 0$ is the risk–free rate and $g$ is the payoff as a function of log–moneyness.
Risk–neutrality of the \emph{discounted} stock requires
\begin{equation}\label{eq:RN-constraint}
\sum_{\alpha\ne 0}\gamma_\alpha\Big(e^{\alpha\Delta x}-1\Big)\;=\;r.
\end{equation}
Indeed, for $s(x)=e^x$ we have $(\mathcal L s)(x)=e^x\sum_\alpha\gamma_\alpha(e^{\alpha\Delta x}-1)$, so
$\frac{d}{d\tau}\mathbb E[s(X_\tau)]=r\,\mathbb E[s(X_\tau)]$ iff \eqref{eq:RN-constraint} holds.

\begin{remark}[Risk–neutral scaling from a shape family]
Choose nonnegative ``shape'' weights $w_\alpha\ge 0$ (supported on $|\alpha|\le A$). Setting
\[
\lambda \;:=\;\frac{r}{\sum_{\alpha\ne 0}w_\alpha\,(e^{\alpha\Delta x}-1)}\,,\qquad
\gamma_\alpha\;:=\;\lambda\,w_\alpha
\]
enforces \eqref{eq:RN-constraint}. This separates economics (martingale) from physics (shape and range of jumps).
\end{remark}

\subsection{Exponential damping and spectral diagonalization}
Direct FFT of \eqref{eq:backward} is hampered by the nonperiodic call payoff.
Introduce exponential damping $v(x,\tau):=e^{-\eta x}u(x,\tau)$ with $\eta>1$.
Then $v$ solves
\begin{equation}\label{eq:damped}
\partial_\tau v \;=\;(\mathcal L^{(\eta)}v)\;-\;r\,v,\qquad 
(\mathcal L^{(\eta)}f)(x):=\sum_{\alpha\ne 0}\gamma_\alpha\Big(e^{\eta\alpha\Delta x}f(x+\alpha\Delta x)-f(x)\Big).
\end{equation}
Let $\widehat{f}(m)$ denote the $N$–point DFT on $[-L,L)$, with discrete frequencies $\omega_m:=2\pi m/N$.
Because $\mathcal L^{(\eta)}$ is convolutional, it is diagonalized by the DFT:
\begin{equation}\label{eq:symbol2}
(\mathcal F\,\mathcal L^{(\eta)}f)(m)\;=\;\psi^{(\eta)}(m)\,\widehat f(m),\qquad
\psi^{(\eta)}(m)\;=\;\sum_{\alpha\ne 0}\gamma_\alpha\Big(e^{\eta\alpha\Delta x}e^{\,i\,\omega_m\alpha}-1\Big).
\end{equation}
Hence the damped solution admits the closed spectral form
\begin{equation}\label{eq:spec-sol}
\widehat v(\tau,m)\;=\;e^{\,\tau\big(\psi^{(\eta)}(m)-r\big)}\,\widehat v(0,m),
\qquad v(\cdot,\tau)=\mathcal F^{-1}\widehat v(\tau,\cdot),\qquad
u(x,\tau)=e^{\eta x}v(x,\tau).
\end{equation}

\subsection{Single–FFT identity for all strikes}
We use the baseline call payoff at unit strike
\[
g_0(x)\;:=\;[e^x-1]_+,
\qquad h(y;\tau)\;:=\;\big(e^{\tau\mathcal L}g_0\big)(y)\,e^{-r\tau}.
\]
By translation–invariance and homogeneity on log–price, for any spot $S_0$ and strike $K$ we have
\begin{equation}\label{eq:single-FFT}
C(S_0,K,\tau)\;=\;K\,h\!\Big(-\ln\frac{K}{S_0}\,;\,\tau\Big).
\end{equation}

\begin{proposition}[Single–FFT pricing for the entire strike grid]
\label{prop:singleFFT}
Fix $\tau>0$. Compute $h(\cdot;\tau)$ once via \eqref{eq:spec-sol} with $g_0$.
Then \eqref{eq:single-FFT} delivers $C(S_0,K,\tau)$ for \emph{all} $K$ by reading $h$ at $-\ln(K/S_0)$.
\end{proposition}

\begin{proof}[Proof (sketch)]
Write $g_K(x)=[e^x-e^{\kappa}]_+=e^{\kappa}[e^{x-\kappa}-1]_+=e^{\kappa}g_0(x-\kappa)$ with $\kappa:=\ln K$.
By translation–invariance, $e^{\tau\mathcal L}g_0(\,\cdot-\kappa)=\big(e^{\tau\mathcal L}g_0\big)(\,\cdot-\kappa)$.
Evaluating at $y=-\ln(K/S_0)$ and discounting gives \eqref{eq:single-FFT}.
\end{proof}

\subsection{Kernel families used for illustration}
We consider two shape families, scaled by the risk–neutral factor in the Remark above:
\begin{itemize}
\item \textbf{Near–field exponential (symmetric):}\quad $w_\alpha=\exp(-|\alpha|/\xi)$, $\xi>0$.
\item \textbf{Skewed heavy tail:}\quad $w_\alpha=(|\alpha|+1)^{-(1+p)}\big(1+\mathrm{skew}\cdot\mathrm{sign}\,\alpha\big)$,
with $p>0$ and $|\mathrm{skew}|<1$ to keep $w_\alpha\ge 0$.
\end{itemize}
These yield, respectively, diffusion–like near–field behavior and fat–tailed, skewed wings.

\subsection{Numerical pipeline and implied volatilities}
Let $g_0(x)=[e^x-1]_+$ and choose a damping $\eta\in(1,2)$ so that $e^{-\eta x}g_0(x)$ decays at both ends.
For prescribed maturities $\tau\in\{\tau_1,\ldots,\tau_J\}$, we compute $h(y;\tau)$ by one FFT per $\tau$,
then read prices for a whole strike grid from \eqref{eq:single-FFT}.
Implied volatilities $\sigma_{\mathrm{impl}}(K,\tau)$ solve 
$C_{\mathrm{BS}}(S_0,K,r,\tau,\sigma)=C(S_0,K,\tau)$ and are recovered by a robust Newton–bisection
update with vega $\mathrm{Vega}=S_0\phi(d_1)\sqrt{\tau}$, where
\[
d_1=\frac{\ln(S_0/K)+(r+\tfrac12\sigma^2)\tau}{\sigma\sqrt{\tau}},\qquad
\phi(z)=\frac{1}{\sqrt{2\pi}}e^{-z^2/2}.
\]

\begin{algorithm}[htbp]
\caption{Single-FFT nonlocal pricing and IV extraction (per maturity $\tau$)}
\label{alg:fft-pricer}

\SetKwInput{KwData}{Input}
\SetKwInput{KwResult}{Output}

\KwData{Grid size $N=2^n$, half-width $L$, spacing $\Delta x=2L/N$; damping $\eta>1$; shape weights $w_\alpha$ on $|\alpha|\le A$; spot $S_0$, rate $r$, strike list $\{K_j\}$.}
\KwResult{Arrays $\{C_j\}$ and $\{\sigma_{\mathrm{impl}}(K_j,\tau)\}$ for the whole strike list.}

\BlankLine

\textbf{Risk-neutral scaling:} $\lambda \leftarrow r\big/\sum_{\alpha\ne 0} w_\alpha \big(e^{\alpha \Delta x}-1\big)$; set $\gamma_\alpha \leftarrow \lambda\, w_\alpha$.\\

\textbf{Initialize (damped) payoff:} On grid $x_k=-L+k\Delta x$, $k=0,\dots,N-1$, set $g_0(x_k)=[e^{x_k}-1]_+$ and $v_0(x_k)\leftarrow e^{-\eta x_k}\, g_0(x_k)$.\\

\textbf{Forward FFT:} $V_0 \leftarrow \mathrm{FFT}(v_0)$.\\

\textbf{Build symbol:} For $m=0,\dots,N-1$ with $\omega_m=2\pi m/N$, compute
\[
\psi^{(\eta)}(m)\leftarrow \sum_{\alpha\ne 0}\gamma_\alpha \Big(e^{\eta \alpha \Delta x} e^{i\omega_m \alpha} - 1\Big).
\]\\

\textbf{Spectral propagation:} $V_\tau(m) \leftarrow \exp\!\big(\tau(\psi^{(\eta)}(m)-r)\big)\, V_0(m)$; then $v(\cdot,\tau)\leftarrow \mathrm{IFFT}(V_\tau)$.\\

\textbf{Undo damping:} $h(x_k;\tau) \leftarrow e^{\eta x_k}\, v(x_k,\tau)$.\\

\textbf{Price \& IV for all strikes:} For each $K_j$, set $y_j \leftarrow -\ln(K_j/S_0)$; read $h(y_j;\tau)$ by periodic interpolation on $\{x_k\}$; set $C_j \leftarrow K_j\, h(y_j;\tau)$; invert to $\sigma_{\mathrm{impl}}(K_j,\tau)$ by Newton--bisection.\\

\end{algorithm}

\subsection{Parameter choices and figures}
Unless otherwise stated we use
\begin{center}
\begin{tabular}{@{}llllll@{}}
\toprule
$S_0$ & $r$ & $N$ & $L$ & $\eta$ & maturities \\
\midrule
$1.0$ & $0.05$ & $4096$ & $5.0$ & $1.5$ & $\tau\in\{0.5,\,2.0\}$ \\
\bottomrule
\end{tabular}
\end{center}
Two kernel families are reported:
\begin{itemize}
\item \textbf{NearFieldExp\_sym:} $w_\alpha=\exp(-|\alpha|/\xi)$ with $\xi=6.0$ (symmetric near–field);
\item \textbf{HeavyTail\_skew:} $w_\alpha=(|\alpha|+1)^{-(1+p)}(1+\mathrm{skew}\cdot\mathrm{sign}\,\alpha)$ with
$p=1.2$ and $\mathrm{skew}=0.35$ (fat tails with positive skew).
\end{itemize}
Figure~\ref{fig:iv-smiles} displays the resulting implied–volatility smiles at two maturities.
Short–maturity smiles are more curved and more sensitive to near–field vs. tail behavior; as the
kernel becomes more near–field and $\Delta x\to 0$, smiles flatten towards the Black–Scholes limit.

\begin{figure}[t]
\centering
\begin{minipage}{0.48\linewidth}
\centering
\includegraphics[width=\linewidth]{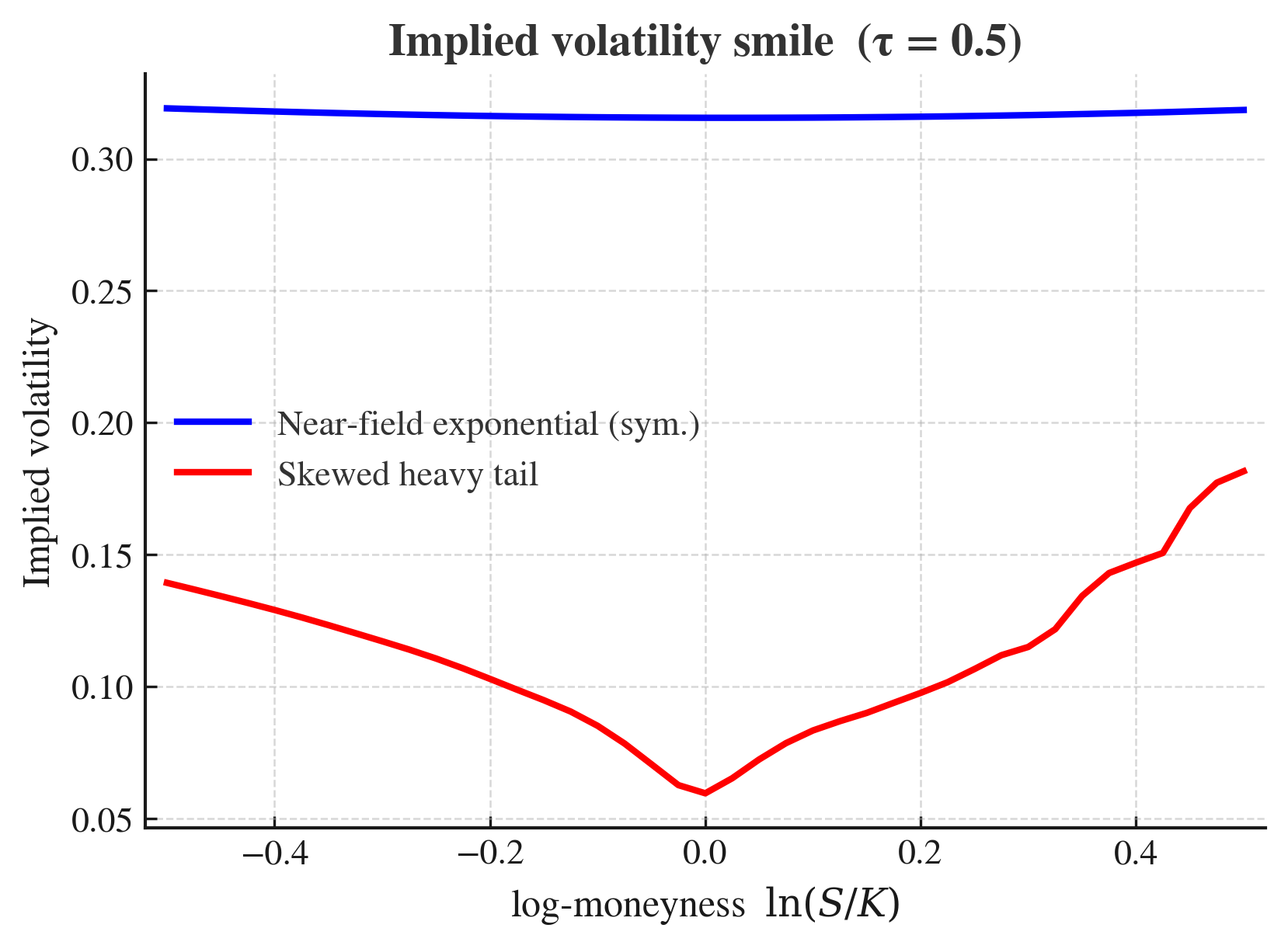}\\[-0.5ex]
{\small $\tau=0.5$}
\end{minipage}\hfill
\begin{minipage}{0.48\linewidth}
\centering
\includegraphics[width=\linewidth]{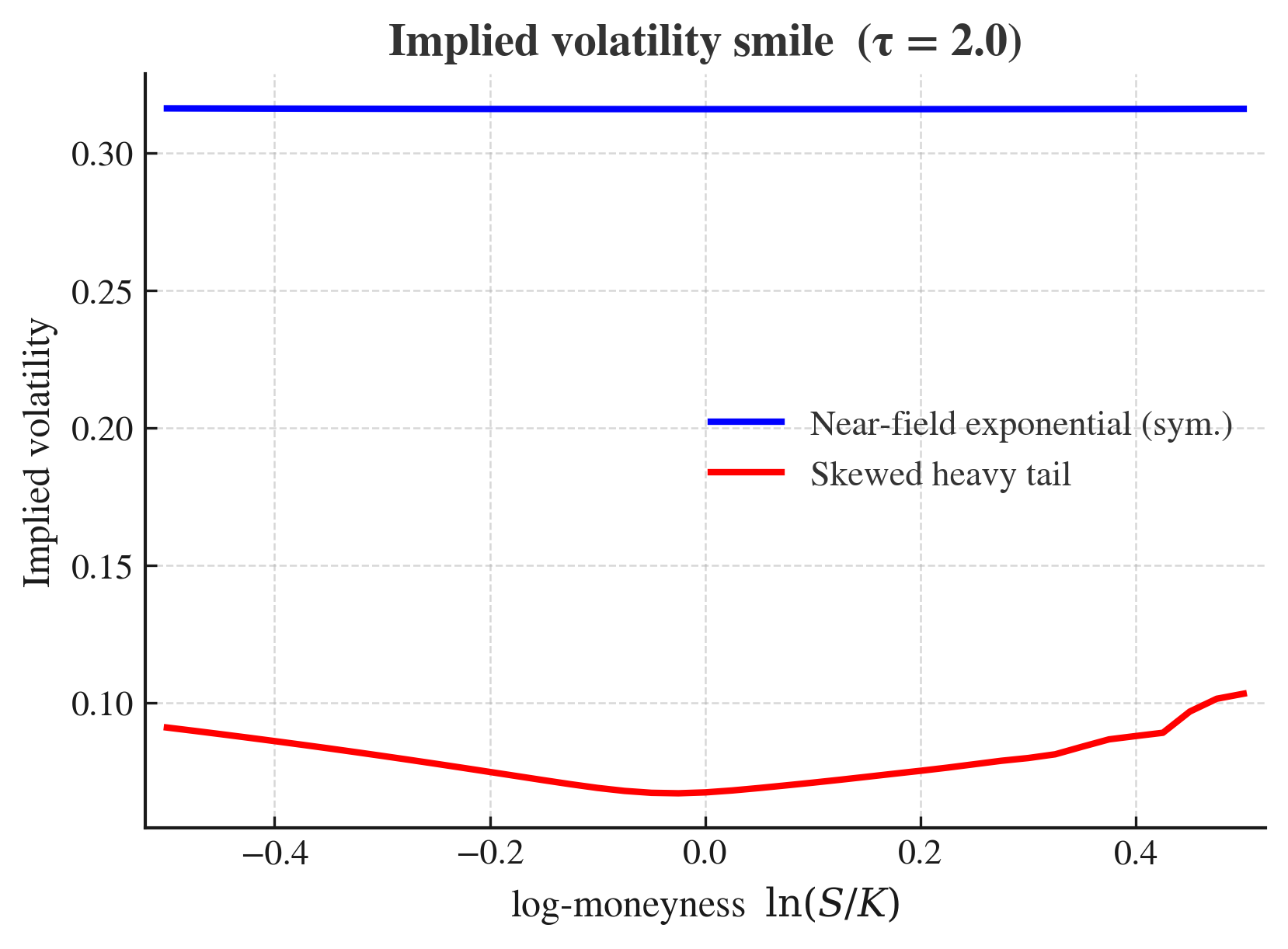}\\[-0.5ex]
{\small $\tau=2.0$}
\end{minipage}
\caption{Implied–volatility smiles from the spectral (FFT/DFT) solver for two kernel families:
NearFieldExp\_sym (symmetric near–field exponential) and HeavyTail\_skew (skewed power–law tail).
Parameters as in the text; each smile is obtained by a \emph{single} FFT per maturity.}
\label{fig:iv-smiles}
\end{figure}

\subsection{Discretization notes and classical limit}
\begin{itemize}
\item \emph{Aliasing/periodicity.} The exponential damping $\eta>1$ makes the damped payoff $e^{-\eta x}g_0(x)$
square–summable and compatible with periodic FFT on $[-L,L)$; $L$ should be large enough so that $h(\cdot;\tau)$ is
flat near the boundaries (empirically checked).
\item \emph{Interpolation.} Reading $h(y;\tau)$ at $y=-\ln(K/S_0)$ uses periodic interpolation on the uniform grid.
\item \emph{Diffusive limit.} As the kernel concentrates near $\alpha=0$ (e.g.\ increasing $\xi$) and $\Delta x\to 0$
with second–moment fixed, the symbol admits a quadratic expansion
$\psi^{(\eta)}(m)\approx \tfrac12\sigma_{\!*}^2(\eta)\,(-\omega_m^2)+\cdots$,
and the prices converge to Black–Scholes; smiles flatten accordingly.
\end{itemize}

\section{Appendix: Axiomatic spectral construction of the one-step transition matrix}
\label{app:spectral-fS}

\paragraph{State space and price basis.}
Let $\mathcal H$ be the Hilbert space $\ell^2(\mathbb Z)$ (infinite price lattice) or
$\ell^2(\mathbb Z/N\mathbb Z)$ (finite cyclic lattice of size $N$). 
We denote by $\{|n\rangle\}$ the \emph{price basis} (price eigenstates), so that
vectors $\psi\in\mathcal H$ are square-summable sequences $\psi(n)=\langle n|\psi\rangle$.

\paragraph{Axioms.}
Fix a time step $\Delta t>0$. We postulate:
\begin{itemize}
\item[(A1)] \textbf{Dynamics.} The one-step evolution is unitary:
$U(\Delta t):\mathcal H\to\mathcal H$ with $U(\Delta t)^\ast U(\Delta t)=I$.
\item[(A2)] \textbf{Hamiltonian form.} There exists a self-adjoint observable $S$ on $\mathcal H$
and a real Borel function $f:\mathbb R\to\mathbb R$ such that the Hamiltonian is
\[
H\ :=\ \hbar\,f(S),\qquad \hbar>0,
\]
and the propagator is the unitary exponential
\[
U(\Delta t)\ =\ e^{-\,i\,H\,\Delta t}\ =\ e^{-\,i\,\hbar\,f(S)\,\Delta t}.
\]
\item[(A3)] \textbf{Born rule in the price basis.} The one-step transition probabilities are
\begin{equation}\label{eq:Born}
P_{nm}(\Delta t)\ :=\ \bigl|\langle n|U(\Delta t)|m\rangle\bigr|^2.
\end{equation}
\end{itemize}

\paragraph{Spectral calculus for $H=\hbar f(S)$.}
By the spectral theorem, there exists a projection-valued measure $E_S(d\lambda)$ such that
\[
S=\int_{\mathbb R}\lambda\,E_S(d\lambda),\qquad
f(S)=\int_{\mathbb R}f(\lambda)\,E_S(d\lambda).
\]
Hence the unitary $U(\Delta t)$ is the spectral multiplier
\begin{equation}\label{eq:U-spectral}
U(\Delta t)\ =\ \int_{\mathbb R} e^{-\,i\,\hbar\,f(\lambda)\,\Delta t}\,E_S(d\lambda).
\end{equation}
For $n,m$ in the price basis, introduce the complex spectral measures
\[
d\mu_{nm}(\lambda)\ :=\ \langle n|E_S(d\lambda)|m\rangle.
\]
Then the matrix elements of $U(\Delta t)$ admit the \emph{explicit spectral-integral formula}
\begin{equation}\label{eq:U-nm}
U_{nm}(\Delta t)
=\langle n|U(\Delta t)|m\rangle
=\int_{\mathbb R} e^{-\,i\,\hbar\,f(\lambda)\,\Delta t}\,d\mu_{nm}(\lambda).
\end{equation}

\paragraph{Explicit transition matrix elements (general form).}
Combining \eqref{eq:Born} and \eqref{eq:U-nm} we obtain
\begin{equation}\label{eq:P-nm-spectral}
\boxed{\quad
P_{nm}(\Delta t)\ =\ \Bigl|\;\int_{\mathbb R} e^{-\,i\,\hbar\,f(\lambda)\,\Delta t}\,
d\mu_{nm}(\lambda)\;\Bigr|^2,\qquad n,m\ \text{in the price basis.}\quad}
\end{equation}
By unitarity of $U(\Delta t)$,
\[
\sum_{n}P_{nm}(\Delta t)=\sum_m P_{nm}(\Delta t)=1,
\]
so $P(\Delta t)$ is (doubly) stochastic.

\paragraph{Translation-invariant (homogeneous) case and Fourier formulas.}
Assume the dynamics are homogeneous on the lattice, i.e.\ $S$ commutes with lattice
translations; equivalently, $S$ is diagonal in the Fourier (plane-wave) basis.
On the infinite lattice, let $|\vartheta\rangle$ ($\vartheta\in[-\pi,\pi)$) be the
generalized eigenstates of the shift, with
\[
\langle n|\vartheta\rangle=(2\pi)^{-1/2}\,e^{\,i\,\vartheta n},\qquad
\langle \vartheta|\vartheta'\rangle=\delta(\vartheta-\vartheta').
\]
If $s(\vartheta)$ denotes the spectral function of $S$ in this basis, then the
\emph{Hamiltonian spectrum (the spectrum of the Hamiltonian)} is
\[
H(\vartheta)\ =\ \hbar\,f\bigl(s(\vartheta)\bigr)\ \in\ \mathbb R,
\]
and the unitary is diagonal in $|\vartheta\rangle$ with eigenvalues
$e^{-iH(\vartheta)\Delta t}$. Transforming back to the price basis yields the
convolution kernel
\begin{equation}\label{eq:U-alpha-Fourier}
U_{nm}(\Delta t)\ =\ u_{n-m}
=\frac{1}{2\pi}\int_{-\pi}^{\pi}
e^{-\,i\,H(\vartheta)\,\Delta t}\,e^{\,i\,\vartheta(n-m)}\,d\vartheta
=\frac{1}{2\pi}\int_{-\pi}^{\pi}
e^{-\,i\,\hbar\,f(s(\vartheta))\,\Delta t}\,e^{\,i\,\vartheta(n-m)}\,d\vartheta.
\end{equation}
Therefore the one-step increment law and transition matrix are
\begin{equation}\label{eq:P-alpha-Fourier}
\boxed{\quad
P_{nm}(\Delta t)
=\bigl|U_{nm}(\Delta t)\bigr|^2
=\left|\frac{1}{2\pi}\int_{-\pi}^{\pi}
e^{-\,i\,\hbar\,f(s(\vartheta))\,\Delta t}\,e^{\,i\,\vartheta(n-m)}\,d\vartheta\right|^2,
\quad\text{(infinite lattice)}\quad}
\end{equation}
which depends only on the increment $\alpha=n-m$.

On the finite cyclic lattice ($n,m\in\{0,\dots,N-1\}$), with discrete frequencies
$\vartheta_r:=2\pi r/N$ and $s_r:=s(\vartheta_r)$, one obtains the discrete Fourier
representation
\begin{equation}\label{eq:P-finite}
\boxed{\quad
P_{nm}(\Delta t)
=\left|\frac{1}{N}\sum_{r=0}^{N-1}
\exp\!\bigl\{-\,i\,\hbar\,f(s_r)\,\Delta t\bigr\}\,e^{\,i\frac{2\pi}{N}r(n-m)}\right|^{\!2},
\quad\text{(finite $N$)}\quad}
\end{equation}
equivalently $U(\Delta t)=F^\ast \mathrm{diag}\!\big(e^{-i\hbar f(s_r)\Delta t}\big)F$ with
$F$ the $N$-point DFT, and $P=|U|^{\circ 2}$ (entrywise squared modulus).

\paragraph{Nontriviality and canonical choice of $S$.}
If $S$ is taken to be the \emph{price operator} (diagonal in $|n\rangle$),
then $U(\Delta t)$ is diagonal with entries $e^{-i\hbar f(s_n)\Delta t}$ and
$P_{nm}(\Delta t)=\delta_{nm}$ (no transitions). To obtain a nontrivial kernel,
$S$ must not commute with the price projections. A canonical choice is the
(discrete) \emph{momentum/shift generator} $\Pi$, for which $s(\vartheta)=\vartheta$
(mod $2\pi$); then $H(\vartheta)=\hbar f(\vartheta)$ and
\eqref{eq:P-alpha-Fourier}--\eqref{eq:P-finite} reduce to the standard
circulant/Fourier formulas used earlier.

\paragraph{Small-time expansion and classical calibration (optional).}
If $f\circ s$ is $C^2$ near $\vartheta=0$,
\[
\hbar f\bigl(s(\vartheta)\bigr)
=\mu\,\vartheta+\tfrac{1}{2}\sigma^2\,\vartheta^2+O(\vartheta^3),
\]
then the curvature $\partial_\vartheta^2(\hbar f\circ s)(0)=\sigma^2$ controls the
one-step increment variance, matching the volatility scale; in a diffusive scaling
limit this recovers the classical (GBM/BSM-type) behavior.

\appendix

\subsection{Appendix. Proof Of Uncertainty Principle}
\label{app:UP}

Assume $\hat{A}$ and $\hat{B}$ are two Hermitian operators, we functional
construct a non-negative integral by introducing a complex parameter
$\xi$:
\begin{equation}
\int_{\Omega}|\xi\hat{A}\Psi+i\hat{B}\Psi|^{2}d\Omega\geq0,\label{eq:integral_inequality}
\end{equation}
and expand this as follows:
\begin{align}
\int_{\Omega}|\xi\hat{A}\Psi+i\hat{B}\Psi|^{2}d\Omega & =\int_{\Omega}(\xi\hat{A}\Psi+i\hat{B}\Psi)^{*}(\xi\hat{A}\Psi+i\hat{B}\Psi)d\Omega\nonumber \\
 & =\int_{\Omega}[\xi^{2}(\hat{A}\Psi)^{*}(\hat{A}\Psi)+i\xi(\hat{A}\Psi)^{*}(\hat{B}\Psi)\nonumber \\
 & -i\xi(\hat{B}\Psi)^{*}(\hat{A}\Psi)+(\hat{B}\Psi)^{*}(\hat{B}\Psi)]d\Omega.\label{eq:expanded_integral}
\end{align}

For the term $\int_{\Omega}\xi^{2}(\hat{A}\Psi)^{*}(\hat{A}\Psi)d\Omega$,
since ${\hat{A}}$ is Hermitian, we obtain

\begin{align}
\int_{\Omega}\xi^{2}(\hat{A}\Psi)^{*}(\hat{A}\Psi)d\Omega & =\xi^{2}(\hat{A}\Psi,\hat{A}\Psi)=\xi^{2}(\hat{A}^{2}\Psi,\Psi)=\xi^{2}(\Psi,\hat{A}^{2}\Psi)\nonumber \\
 & =\xi^{2}\int_{\Omega}\Psi^{*}\hat{A}^{2}\Psi d\Omega,\label{eq:A_term}
\end{align}
and similarly for the term $\int_{\Omega}(\hat{B}\Psi)^{*}(\hat{B}\Psi)d\Omega$.
Note that

\begin{align}
(\Psi,\hat{A}\hat{B}\Psi)-(\Psi,\hat{B}\hat{A}\Psi) & =\int_{\Omega}\Psi^{*}(\hat{A}\hat{B}\Psi)d\Omega-\int_{\Omega}\Psi^{*}(\hat{B}\hat{A}\Psi)d\Omega\nonumber \\
 & =\int_{\Omega}\Psi^{*}[\hat{A},\hat{B}]\Psi d\Omega=i\int_{\Omega}\Psi^{*}\hat{C}\Psi d\Omega=i\overline{C_{\Psi}},\label{eq:commutator_result}
\end{align}
with $i\hat{C}=[\hat{A},\hat{B}]$. Therefore

\begin{align}
\int_{\Omega}|\xi\hat{A}\Psi+i\hat{B}\Psi|^{2}d\Omega & =\xi^{2}(\Psi,\hat{A}^{2}\Psi)+(\Psi,\hat{B}^{2}\Psi)+i\xi(i\overline{C_{\Psi}})\nonumber \\
 & =\xi^{2}\overline{A_{\Psi}^{2}}+\overline{B_{\Psi}^{2}}-\xi\overline{C_{\Psi}}=F(\xi)\geq0,\label{eq:final_inequality}
\end{align}
where

\begin{equation}
\begin{aligned}\overline{A_{\Psi}^{2}} & =(\Psi,\hat{A}^{2}\Psi)=(\hat{A}\Psi,\hat{A}\Psi)=\int_{\Omega}(\hat{A}\Psi)^{*}(\hat{A}\Psi)d\Omega.\end{aligned}
\label{eq:A_squared_expectation}
\end{equation}

Given $F(\xi)\geq0$ for all $\xi$, we can find its minimum:

\begin{equation}
\arg\min_{\xi}F(\xi)=\frac{\overline{C_{\Psi}}}{2\overline{A_{\Psi}^{2}}},\label{eq:min_xi}
\end{equation}
therefore

\begin{align}
\min F(\xi) & =\left(\frac{\overline{C_{\Psi}}}{2\overline{A_{\Psi}^{2}}}\right)^{2}\overline{A_{\Psi}^{2}}+\overline{B_{\Psi}^{2}}-\left(\frac{\overline{C_{\Psi}}}{2\overline{A_{\Psi}^{2}}}\right)\overline{C_{\Psi}}\nonumber \\
 & =\frac{\overline{C_{\Psi}}^{2}}{4(\overline{A_{\Psi}^{2}})^{2}}\overline{A_{\Psi}^{2}}+\overline{B_{\Psi}^{2}}-\frac{\overline{C_{\Psi}}^{2}}{2\overline{A_{\Psi}^{2}}}=\overline{B_{\Psi}^{2}}-\frac{1}{4}\frac{\overline{C_{\Psi}}^{2}}{\overline{A_{\Psi}^{2}}}\geq0,\label{eq:min_F_xi}
\end{align}
this implies

\begin{equation}
\overline{B_{\Psi}^{2}}\geq\frac{1}{4}\frac{\overline{C_{\Psi}}^{2}}{\overline{A_{\Psi}^{2}}}\Rightarrow\overline{A_{\Psi}^{2}}\cdot\overline{B_{\Psi}^{2}}\geq\frac{1}{4}(\overline{C_{\Psi}})^{2}.\label{eq:uncertainty_inequality}
\end{equation}

Let's introduce $\tilde{A}=(\hat{A}-\overline{A_{\Psi}})\hat{I}$
and $\tilde{B}=(\hat{B}-\overline{B_{\Psi}})\hat{I}$, then

\begin{equation}
[\tilde{A},\tilde{B}]=[(\hat{A}-\overline{A_{\Psi}})\hat{I},(\hat{B}-\overline{B_{\Psi}})\hat{I}]=[\hat{A},\hat{B}]=i\hat{C},\label{eq:commutator_tilde}
\end{equation}
we've already shown: $[\hat{A},\hat{B}]=i\hat{C}\Rightarrow\overline{A_{\Psi}^{2}}\cdot\overline{B_{\Psi}^{2}}\geq\frac{1}{4}(\overline{C_{\Psi}})^{2}$,
therefore

\begin{equation}
\overline{\tilde{A}_{\Psi}^{2}}\cdot\overline{\tilde{B}_{\Psi}^{2}}\geq\frac{1}{4}(\overline{C_{\Psi}})^{2}\Rightarrow(\hat{A}-\overline{A_{\Psi}})_{\Psi}^{2}\cdot(\hat{B}-\overline{B_{\Psi}})_{\Psi}^{2}\geq\frac{1}{4}\overline{C_{\Psi}}^{2}\neq0.
\end{equation}
\appendix

\end{document}